\newif\ifsubmission
\newif\ifcomments
\newif\ifanonymous
\newif\ifshaphered

\ifdefined\isshaphered
\shapheredtrue
\fi

\ifdefined\issubmission
\submissiontrue
\fi

\ifdefined\isanonymous
\anonymoustrue
\fi

\ifdefined\iscomments
\commentstrue
\fi

\ifsubmission
\makeatletter
\else
\makeatletter
\def\@copyrightspace{\relax}
\makeatother
\fi

\documentclass[10pt,journal,compsoc]{IEEEtran}

\ifCLASSOPTIONcompsoc
  \usepackage[nocompress]{cite}
\else
  \usepackage{cite}
\fi

\ifCLASSINFOpdf
   \usepackage[pdftex]{graphicx}
\else
   \usepackage[dvips]{graphicx}
\fi

\usepackage{algorithm}
\usepackage{algorithmicx}
\usepackage{algpseudocode}
\usepackage{mdframed}
\usepackage{times}
\usepackage{subfigure}
\usepackage{tikz}
\usepackage{caption}

\usepackage{color}
\usetikzlibrary{arrows}
\usepackage{pgfplots, pgfplotstable}
\pgfplotsset{compat=newest}
\definecolor{carmine}{RGB}{150, 0, 24}
\definecolor{forest}{RGB}{34, 139, 34}
\definecolor{blues4}{RGB}{49, 130, 189}
\pgfplotsset{grid style={dashed,gray}}
\usepackage{paralist}
\renewcommand{\itemize}{\compactitem}
 \renewcommand{\enumerate}{\compactenum}

\usepackage[font={small}]{caption}
\usepackage{longtable}

\usepackage{amsmath}
\usepackage{pifont}
\newcommand{\xmark}{\text{\ding{55}}}

\usepackage{caption}
\usepackage{url}
\usepackage{hyperref}

\usepackage{tabularx}

\usepackage[normalem]{ulem}
\usepackage{listings}
\usepackage{textcomp}
\usepackage{mathpazo}
\usepackage{tgpagella}
\usepackage{amsthm}

\newcommand{\remove}[1]{}

\newtheorem{theo}{\bf Theorem}
\newtheorem{lemma}{\bf Lemma}

\newcommand*{\textlabel}[2]{%
  \edef\@currentlabel{#1}
  \phantomsection
  #1\label{#2}
}\usepackage{hyperref}
\makeatletter

\renewcommand{\footnotesize}{\small}

\newcommand\changed[1]{{#1}}
\newcommand\changedagain[1]{{#1}}
\newcommand\ICchanged[1]{{#1}}
\newcommand\ICchangedagain[1]{{#1}}

\newcommand\revision[1]{#1}

\ifcomments
\newcommand\ghassan[1]{\textcolor{brown}{Ghassan: #1}}
\newcommand\wenting[1]{\textcolor{purple}{Wenting: #1}}
\newcommand\jianliu[1]{\textcolor{green}{Jian Liu: #1}}
\newcommand\asokan[1]{\textcolor{orange}{Asokan: #1}}
\newcommand\TODO[1]{\textcolor{red}{TODO: #1}}
\else
\newcommand\ghassan[1]{}
\newcommand\wenting[1]{}
\newcommand\jianliu[1]{}
\newcommand\asokan[1]{}
\newcommand\TODO[1]{}
\fi

\newcommand{\fig}[1]{Fig.~\ref{#1}}

\usepackage{xspace}
\newcommand{\BFT}{\protect{FastBFT}\xspace}
\newcommand{\Ser}{\ensuremath{\mathcal{S}}}
\newcommand{\Cli}{\ensuremath{\mathcal{C}}}

\newcommand{\prepareP}{\textsf{prepare}\xspace}
\newcommand{\commitP}{\textsf{commit}\xspace}
\newcommand{\preprepareP}{\textsf{pre-prepare}\xspace}
\newcommand{\replyP}{\textsf{reply}\xspace}
\newcommand{\preprocessingP}{\textsf{pre-processing}\xspace}

\newcommand{\request}{\textrm{REQUEST}\xspace}
\newcommand{\reply}{\textrm{REPLY}\xspace}

\newcommand{\prepare}{\textrm{PREPARE}\xspace}
\newcommand{\commit}{\textrm{COMMIT}\xspace}

\newcommand{\reqviewchange}{\textrm{REQ-VIEW-CHANGE}\xspace}
\newcommand{\viewchange}{\textrm{VIEW-CHANGE}\xspace}
\newcommand{\newview}{\textrm{NEW-VIEW}\xspace}

\newcommand{\suspect}{\textrm{SUSPECT}\xspace}
\newcommand{\newtree}{\textrm{NEW-TREE}\xspace}
\newcommand{\rejoin}{\textrm{REJOIN}\xspace}
\newcommand{\TA}{\textrm{TEE}}

\newcommand{\E}{\textsf{E}}
\newcommand{\D}{\textsf{D}}
\newcommand{\Enc}{\textsf{Enc}}
\newcommand{\Dec}{\textsf{Dec}}
\newcommand{\Sign}{\textsf{Sign}}
\newcommand{\Verify}{\textsf{Vrfy}}

\newcommand{\cipher}{\varrho}

\begin{document}
%
\title{Scalable Byzantine Consensus via Hardware-assisted Secret Sharing}

\ifanonymous
\else
\author{Jian~Liu, 
            Wenting~Li,
            Ghassan~O.~Karame,~\IEEEmembership{Member,~IEEE,}
        and~N.~Asokan,~\IEEEmembership{Fellow,~IEEE}
\IEEEcompsocitemizethanks{
\IEEEcompsocthanksitem Jian Liu and N. Asokan are with the Department
of Computer Science, Aalto University, Finland. E-mail: jian.liu@aalto.fi, asokan@acm.org \protect\\
\IEEEcompsocthanksitem Wenting Li and Ghassan O. Karame are with NEC Laboratories Europe, Germany. 
E-mail: \{wenting.li, ghassan.karame\}@neclab.eu
}
}

\remove{
\author{\IEEEauthorblockN{Jian Liu}
\IEEEauthorblockA{
Aalto University, Finland\\
jian.liu@aalto.fi}
\and
\IEEEauthorblockN{Wenting Li}
\IEEEauthorblockA{NEC Laboratories Europe\\
wenting.li@neclab.eu}
\and
\IEEEauthorblockN{Ghassan O. Karame}
\IEEEauthorblockA{NEC Laboratories Europe\\
ghassan@karame.org}
\and
\IEEEauthorblockN{N. Asokan}
\IEEEauthorblockA{Aalto University, Finland\\
asokan@acm.org}
}
}

\remove{
\ifsubmission
\IEEEoverridecommandlockouts
\makeatletter\def\@IEEEpubidpullup{9\baselineskip}\makeatother
\IEEEpubid{\parbox{\columnwidth}{
}
\hspace{\columnsep}\makebox[\columnwidth]{}}
\fi
}

\markboth{\textcopyright~2018 IEEE DOI:10.1109/TC.2018.2860009}%
{Shell \MakeLowercase{\textit{et al.}}: Bare Demo of IEEEtran.cls for Computer Society Journals}

\IEEEtitleabstractindextext{%
\begin{abstract}
 The surging interest in blockchain technology has revitalized the search for effective Byzantine consensus schemes. In particular, the blockchain community has been looking for ways to effectively integrate traditional Byzantine fault-tolerant (BFT) protocols 
 into a blockchain consensus layer allowing various financial institutions to securely agree on the order of transactions.
 However, existing BFT protocols can only scale to tens of nodes due to their $O(n^2)$ message complexity.
 
In this paper, we propose \BFT, \revision{\textlabel{a fast and scalable BFT protocol}{R2(1.1)}}. At the heart of \BFT is a novel message aggregation technique that combines hardware-based trusted execution environments (TEEs) with lightweight secret sharing. 
Combining this technique with several other optimizations (i.e., optimistic execution, tree topology and failure detection),  \BFT achieves low latency and high throughput even for large scale networks.
Via systematic analysis and experiments, we demonstrate that \BFT has better scalability and performance than previous BFT protocols.


\end{abstract}

}

\maketitle

\ifshaphered
\newcommand{\on}[2]{{\textcolor{red}{#1}} \textcolor{blue}{#2}}
\input{shepherd-changelog}
\else
\newcommand{\on}[2]{#2}
\fi

\section{Introduction}
\label{sec:intro}

Byzantine fault-tolerant (BFT) protocols have not yet seen significant real-world deployment.
\revision{\textlabel{There are several potential reasons}{R2(3)} for this including the poor efficiency and scalability of current BFT protocols and, more importantly, due to the fact that often Byzantine faults are not perceived to be a major concern in well-maintained data centers.}
Consequently,
existing commercial systems like those in Google~\cite{180268} and Amazon~\cite{Verbitski:2017:AAD:3035918.3056101} rely on weaker crash fault-tolerant variants (e.g., Paxos~\cite{Lamport:1998:PP:279227.279229} and Raft~\cite{Raft}).

Recent interest in blockchain technology has given fresh impetus for BFT protocols.
A blockchain is a key enabler for {\em distributed consensus}, serving as a public ledger for digital currencies (e.g., Bitcoin) and other applications.
Bitcoin's blockchain relies on the well-known proof-of-work (PoW) mechanism to ensure probabilistic consistency guarantees on the order and correctness of transactions.
PoW currently accounts for more than 90\% of the total market share of existing digital currencies.
\ifsubmission
\else
(e.g., Bitcoin, Litecoin, DogeCoin, Ethereum)
\fi
However, Bitcoin's PoW has been severely criticized for its
considerable waste of energy and meagre 
transaction throughput ($\sim$7 transactions per second)~\cite{DBLP:conf/ccs/GervaisKWGRC16}.

To remedy these limitations, researchers and practitioners are investigating integration of BFT protocols with blockchain consensus\ifsubmission. \else
 to enable financial institutions and supply chain management partners to agree on the order and correctness of exchanged information.
\fi
This represents the first opportunity for BFT protocols to be integrated into real-world systems. For example, IBM's Hyperledger/Fabric blockchain~\cite{Hyperledger} currently relies on PBFT~\cite{PBFT} for consensus.
\changedagain{
While PBFT can achieve higher throughput 
than Bitcoin's consensus layer
~\cite{Fabric}, it cannot match, by far, the transactional volumes of existing payment methods (e.g., Visa handles tens of thousands of transactions per second~\cite{Visa}).
Furthermore, PBFT only scales to few tens of nodes, \revision{since it needs to exchange $O(n^2)$ messages to reach consensus on a single operation among $n$ servers~\cite{PBFT}.}}
Thus, enhancing the scalability and performance of BFT protocols is essential for ensuring their practical deployment in existing industrial blockchain solutions.

In this paper, we propose \BFT, \revision{\textlabel{a fast and scalable BFT protocol}{R2(1.2)}}. At the heart of \BFT is a novel {\em message aggregation} technique that combines hardware-based \emph{trusted execution environments} (e.g., Intel SGX) with lightweight 
secret sharing.
Aggregation reduces message complexity from $O(n^2)$ to $O(n)$~\cite{CoSi}.
Unlike previous schemes, message aggregation in \BFT does \emph{not} require any public-key operations (e.g.,  multisignatures),
thus incurring considerably lower  computation/communication overhead.
\BFT further balances
\ifsubmission
\else
computation and communication
\fi
load by arranging nodes in a tree topology,
so that inter-server communication and message aggregation take place along edges of the tree.
\revision{\BFT \textlabel{adopts the}{R2(5.1)} \emph{optimistic} BFT paradigm~\cite{Distler2016} that 
only requires a subset of nodes to {\em actively} run the protocol. }
Finally, we use a simple \emph{failure detection} mechanism that makes it possible for FastBFT to deal with non-primary faults efficiently.

Our experiments show that, 
the throughput of \BFT is 
significantly larger compared to other BFT protocols we evaluated~\cite{Zyzzyva, MinBFT, CheapBFT}. 
As the number of nodes increases, \BFT exhibits considerably slower decline in throughput compared to other BFT protocols.
This makes \BFT an ideal consensus layer candidate for next-generation blockchain systems ---
e.g., 
assuming 1~MB blocks and 250 byte transaction records (as in Bitcoin), \BFT can process over 100,000 transactions per second.



In \BFT, we made specific design choices as to how the building blocks (e.g., message aggregation technique, or communication topology) are selected and used.
Alternative design choices would yield different BFT variants featuring various tradeoffs between efficiency and resilience. We capture this tradeoff through a framework that compares such variants.

In summary, we make the following contributions:
\begin{itemize}
\item We propose \BFT, \revision{\textlabel{a  fast and scalable
    BFT protocol}{R2(1.3)}} (Sections~\ref{sec:overview} and \ref{sec:BFT}). 
\item We describe a framework that captures a set of important design choices and allows us to situate \BFT in the context of a number of possible BFT variants (both previously proposed and novel variants) (Section~\ref{sec:discussion}).
\item We present a full implementation of \BFT and a systematic performance analysis comparing \BFT with several BFT variants. Our results show that \BFT outperforms other variants in terms of efficiency (latency and throughput) and scalability (Section~\ref{sec:performance}).
\end{itemize}

\section{Preliminaries}
\label{sec:Background}


In this section, we describe the problem we tackle, outline known BFT protocols and existing optimizations.


\subsection{State Machine Replication (SMR)}
\label{sec:smr}

SMR~\cite{SMR} is a distributed computing primitive for implementing fault-tolerant services where the state of the system is replicated across different nodes, called ``replicas'' ($\Ser$s). Clients ($\Cli$s) send requests to  $\Ser$s, which are expected to execute the same order of requested operations (i.e., maintain a common state). 
However, some  $\Ser$s may be faulty and their failure mode can be either {\em crash} or {\em Byzantine} (i.e., deviating arbitrarily from the protocol~\cite{Byzantine}).
Fault-tolerant SMR must ensure two \emph{correctness} guarantees:
\begin{itemize}
\item {\em Safety}: all non-faulty replicas execute the requests in the same order (i.e., consensus), and
\item {\em Liveness}: clients eventually receive replies to their requests.
\end{itemize}
\revision{Fischer-Lynch-Paterson (FLP) impossibility~\cite{FLP} proved that fault-tolerance {\em cannot} be deterministically achieved in an asynchronous communication model where \textlabel{no bounds on transmission delays can be assumed.}{R1(14)}}

\subsection{Practical Byzantine Fault Tolerance (PBFT)}
\label{subsec:PBFT}

For decades, researchers have been struggling to circumvent the FLP impossibility.
One approach, PBFT~\cite{PBFT}, leverage the {\em weak synchrony} assumption under which
messages are guaranteed to be delivered after a certain time bound.
%

%
One replica, the {\em primary} $\Ser_p$, decides the order for clients' requests, and forwards them to other replicas $\Ser_i$s.
Then, {\em all} replicas together run a three-phase (\preprepareP/\prepareP/\commitP) agreement protocol to agree on the order of requests. 
Each replica then processes each request and sends a response to the corresponding client. 
\changed{The client accepts the result only if it has received at least $f+1$ consistent replies.}
We refer to BFT protocols incorporating such message patterns (\fig{fig:PBFT}) as {\em classical} BFT.
$\Ser_p$ may become faulty: either stop processing requests (crash) or send contradictory messages to different $\Ser_i$s (Byzantine). The latter is referred to as {\em equivocation}.
On detecting that $\Ser_p$ is faulty,  $\Ser_i$s trigger a {\em view-change} to select a new primary.
The weak synchrony assumption guarantees that 
\mbox{view-change will eventually succeed}.

\subsection{Optimizing for the Common Case}


Since agreement in classical BFT is expensive, prior works have attempted to improve performance based on the fact that replicas rarely fail.
We group these efforts into two categories:


\noindent{\bf Speculative.}
Kotla et al. present Zyzzyva~\cite{Zyzzyva} that uses speculation to improve performance.
Unlike classical BFT, $\Ser_i$s in Zyzzyva execute $\Cli$s' requests following the order proposed by $\Ser_p$, {\em without} running any explicit agreement protocol.
After execution is completed, all replicas reply 
to $\Cli$.
If $\Ser_p$ equivocates, $\Cli$ will receive inconsistent replies.
In this case,  $\Cli$ helps correct replicas to recover from their inconsistent states to a common state.
Zyzzyva can reduce the overhead of state machine replication to near optimal. 
We refer to BFT protocols following this message pattern as {\em speculative} BFT.

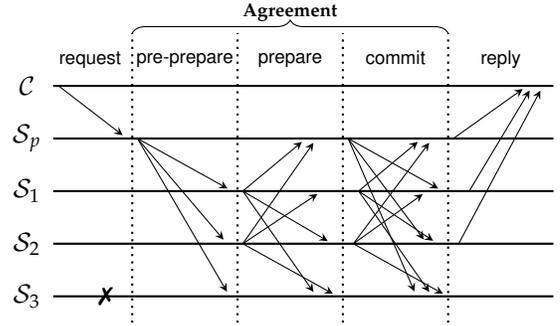
\begin{figure}[tbp]
\centering
\begin{tikzpicture}
  [scale=.7,auto=center]
   \node at (1.2, 5.5) {\scriptsize \sf request};
   \node at (3, 5.5) {\scriptsize \sf pre-prepare};
   \node at (5, 5.5) {\scriptsize \sf prepare};
    \node at (7, 5.54) {\scriptsize \sf commit};
     \node at (9, 5.5) {\scriptsize \sf reply};

   \draw [thick, dotted] (2, 0.5) -- (2, 6);
    \draw [thick, dotted] (4, 0.5) -- (4, 6);
     \draw [thick, dotted] (6, 0.5) -- (6, 6);
      \draw [thick, dotted] (8, 0.5) -- (8, 6);

   \node at (0,5) {$\Cli$};
   \draw [thick] (0.5,5) -- (10,5);
   \draw [->,shorten >=2pt,>=stealth] (0.6, 5) -- (1.9, 4);

    \node at (0,4) {$\Ser_p$};
     \draw [thick] (0.5,4) -- (10,4);
     \draw [->,shorten >=2pt,>=stealth] (2.1, 4) -- (3.9, 3);
     \draw [->,shorten >=2pt,>=stealth] (2.1, 4) -- (3.8, 2);
     \draw [->,shorten >=2pt,>=stealth] (2.1, 4) -- (3.85, 1);

     \draw [->,shorten >=2pt,>=stealth] (6.1, 4) -- (7.9, 3);
     \draw [->,shorten >=2pt,>=stealth] (6.1, 4) -- (7.6, 2);
     \draw [->,shorten >=2pt,>=stealth] (6.1, 4) -- (7.4, 1);

     \draw [->,shorten >=2pt,>=stealth] (8.1, 4) -- (9.5, 5);

   \node at (0,3) {$\Ser_1$};
   \draw [thick] (0.5,3) -- (10,3);
     \draw [->,shorten >=2pt,>=stealth] (4.1, 3) -- (5.3, 4);
     \draw [->,shorten >=2pt,>=stealth] (4.1, 3) -- (5.8, 2);
     \draw [->,shorten >=2pt,>=stealth] (4.1, 3) -- (5.5, 1);

     \draw [->,shorten >=2pt,>=stealth] (6.3, 3) -- (7.5, 4);
     \draw [->,shorten >=2pt,>=stealth] (6.3, 3) -- (7.8, 2);
     \draw [->,shorten >=2pt,>=stealth] (6.3, 3) -- (7.7, 1);

      \draw [->,shorten >=2pt,>=stealth] (8.4, 3) -- (9.6, 5);

    \node at (0,2) {$\Ser_2$};
  \draw [thick] (0.5,2) -- (10,2);
     \draw [->,shorten >=2pt,>=stealth] (4.1, 2) -- (5.5, 4);
     \draw [->,shorten >=2pt,>=stealth] (4.1, 2) -- (5.6, 3);
     \draw [->,shorten >=2pt,>=stealth] (4.1, 2) -- (5.9, 1);

     \draw [->,shorten >=2pt,>=stealth] (6.2, 2) -- (7.7, 4);
     \draw [->,shorten >=2pt,>=stealth] (6.2, 2) -- (7.6, 3);
     \draw [->,shorten >=2pt,>=stealth] (6.2, 2) -- (7.99, 1);

     \draw [->,shorten >=2pt,>=stealth] (8.2, 2) -- (9.8, 5);

    \node at (0,1) {$\Ser_3$};
    \draw [thick] (0.5,1) -- (10,1);
    \node at (1.5, 1) {$\xmark$};

\draw [decorate,decoration={brace,amplitude=5pt}]
(2,6) -- (8,6) node [black,midway,yshift=8]
{\scriptsize {\bf Agreement}};

\end{tikzpicture}
        \caption{Message pattern in PBFT.}
        \label{fig:PBFT}
\end{figure}

\noindent{\bf Optimistic.}
\revision{\textlabel{Distler et al. proposed}{R2(5.2)} a resource-efficient BFT (ReBFT) replication architecture~\cite{Distler2016}.}
In the common case, only a subset of replicas are required to run the agreement protocol. Other replicas passively update their states and become actively involved only in case the agreement protocol fails.
We call BFT protocols following this message pattern as {\em optimistic} BFT.
Notice that such protocols are different from speculative BFT in which explicit agreement is \emph{not} required in the common case.

\subsection{Using Hardware Security Mechanisms}

Hardware security mechanisms have become widely available on commodity computing platforms.
Trusted execution environments (TEEs) are already pervasive on mobile platforms~\cite{TEE}.
Newer TEEs such as Intel's SGX~\cite{SGX1, SGX2} are being deployed on PCs and servers.
TEEs provide protected memory and isolated execution so that 
the regular operating system or applications \revision{\textlabel{can neither control nor observe}{R1(15)} the data being stored or processed inside them.}
TEEs also allow remote verifiers to ascertain the current configuration and behavior of a device via \emph{remote attestation}. 
\revision{\textlabel{In other words}{R2(6)}, TEE can only crash but not be Byzantine.}

Previous work showed how to use hardware security to reduce the number of replicas and/or communication phases for BFT protocols~\cite{Correia2005, Chun07, TrInc, MinBFT, EBAWA, CheapBFT}.
For example, MinBFT~\cite{MinBFT} improves PBFT using a {\em trusted counter service} to prevent equivocation \revision{\cite{Chun07}} by faulty replicas.
Specifically, each replica's local TEE maintains a unique, monotonic and sequential counter;
each message is required to be bound to a unique counter value.
Since monotonicity of the counter is ensured by TEEs, replicas cannot assign the same counter value to different messages.
As a result, the number of required replicas is reduced from $3f+1$ to $2f+1$ (where $f$ is the maximum number of tolerable faults) and the number of communication phases is reduced from 3 to 2 (\prepareP/\commitP).
Similarly, MinZyzzyva uses TEEs to reduce the number of replicas in Zyzzyva but requires the same number of communication phases~\cite{MinBFT}.
CheapBFT~\cite{CheapBFT} uses TEEs in an optimistic BFT protocol.
In the absence of faults, CheapBFT requires only $f+1$ active replicas to agree on and execute client requests. 
The other $f$ passive replicas just modify their states by processing state updates provided by the active replicas. 
In case of suspected faulty behavior, CheapBFT triggers a transition protocol to activate passive replicas, and then switches to MinBFT.

\subsection{Aggregating Messages}


Agreement in BFT requires each $\Ser_i$ to multicast a commit message to all (active) replicas to signal that it agrees with the order proposed by $\Ser_p$.
This leads to $O(n^2)$ message complexity (\fig{fig:PBFT}).
A natural solution is to use {\em message aggregation} techniques to combine messages from multiple replicas.
By doing so, each $\Ser_i$ only needs to send and receive a single message.
For example, collective signing (CoSi)~\cite{CoSi} relies on {\em multisignatures} to aggregate messages.
It was used by ByzCoin~\cite{ByzCoin} to improve scalability of PBFT. Multisignatures allow multiple signers to produce a compact, joint signature on common input. Any verifier that holds the aggregate public key can verify the signature in constant time. However, multisignatures generally require larger message sizes and longer processing~times.

\section{\BFT Overview}
\label{sec:overview}

In this section, we give an overview of \BFT before providing a detailed specification in Section~\ref{sec:BFT}.

\noindent\textbf{System model.} \BFT operates in the same setting as in Section~\ref{subsec:PBFT}: it guarantees safety in asynchronous networks but requires weak synchrony for liveness. 
We further assume that each replica holds a hardware-based TEE that maintains a monotonic counter
and a rollback-resistant memory\footnote{\scriptsize{Rollback-resistant memory can be built via monotonic counters~\cite{dirtybit}.}}.
TEEs can verify one another using remote attestation and establish secure communication channels among them~\cite{anati2013innovative}. 
We assume that faulty replicas may be Byzantine but TEEs may only crash.


\noindent\textbf{Strawman design.} We choose the optimistic paradigm \revision{(like CheapBFT~\cite{CheapBFT})} where $f+1$ active replicas agree and execute the requests and the other $f$ passive replicas just update their states. 
The optimistic paradigm achieves a strong tradeoff between efficiency and resilience (see Section~\ref{sec:discussion}).
We use \textbf{message aggregation} \revision{(\textlabel{with one more communication step}{R2(8.1)})} to reduce message complexity to $O(n)$: during \commitP, each active replica $\Ser_i$  sends its commit message directly to the primary $\Ser_p$ instead of multicasting to all replicas. To avoid the overhead associated with message aggregation using primitives like multisignatures, we use \textbf{secret sharing} for aggregation. 
\revision{\textlabel{An essential assumption}{R1:7} of our protocol is that secrets are one-time.}
To facilitate this, we introduce an additional \preprocessingP phase in the design of \BFT.  \fig{fig:FastBFT} depicts the overall message pattern of \BFT.

First, consider the following strawman design.
During \preprocessingP, $\Ser_p$ generates a set of random secrets and publishes the cryptographic hash of each secret.
Then, $\Ser_p$ splits each secret into shares and sends one share to each active $\Ser_i$.
Later, during \prepareP,  $\Ser_p$ binds each client request to 
a previously shared secret.
During \commitP, each active $\Ser_i$ signals its commitment by revealing its share of the secret.
$\Ser_p$ gathers all such shares to reconstruct the secret, which represents the aggregated commitment of all replicas.
$\Ser_p$ multicasts the reconstructed secret to all  active $\Ser_i$s which can verify it with respect to the corresponding hash.
During \replyP, the same approach is used to aggregate reply messages from all \revision{active} $\Ser_i$:
after verifying the secret, 
$\Ser_i$ reveals its share of the next secret to $\Ser_p$ which reconstructs the reply secret and returns it to the client as well as to all passive replicas.
Thus, the client and passive replicas only need to receive one reply instead of $f+1$.
\revision{$\Ser_p$ \textlabel{includes the two opened secrets}{R(1)9} and their hashes (which are published in the \preprocessingP phases) in the reply messages.}

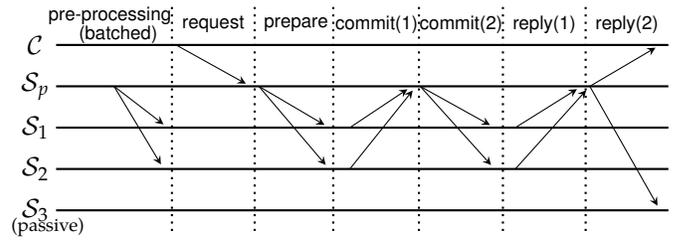
\begin{figure}[tbp]
\centering
\begin{tikzpicture}
  [scale=.55,auto=center]
   \node at (0.5, 5.7) {\scriptsize \sf pre-processing};
    \node at (0.7, 5.25) {\scriptsize \sf (batched)};
   \node at (3, 5.5) {\scriptsize \sf request};
   \node at (5, 5.5) {\scriptsize \sf prepare};
    \node at (6.95, 5.5) {\scriptsize \sf commit(1)};
    \node at (8.99, 5.5) {\scriptsize \sf commit(2)};
     \node at (11, 5.5) {\scriptsize \sf reply(1)};
     \node at (13, 5.5) {\scriptsize \sf reply(2)};

   \draw [thick, dotted] (2, 0.5) -- (2, 6);
    \draw [thick, dotted] (4, 0.5) -- (4, 6);
     \draw [thick, dotted] (5.9, 0.5) -- (5.9, 6);
      \draw [thick, dotted] (7.97, 0.5) -- (7.97, 6);
      \draw [thick, dotted] (10, 0.5) -- (10, 6);
       \draw [thick, dotted] (12, 0.5) -- (12, 6);

   \node at (-1.3,5) {$\Cli$};
   \draw [thick] (-0.8,5) -- (14,5);
   \draw [->,shorten >=2pt,>=stealth] (2.1, 5) -- (3.9, 4);

    \node at (-1.3,4) {$\Ser_p$};
     \draw [thick] (-0.8,4) -- (14,4);
     \draw [->,shorten >=2pt,>=stealth] (0.6, 4) -- (1.9, 3);
     \draw [->,shorten >=2pt,>=stealth] (0.6, 4) -- (1.8, 2);

     \draw [->,shorten >=2pt,>=stealth] (4.1, 4) -- (5.8, 3);
     \draw [->,shorten >=2pt,>=stealth] (4.1, 4) -- (5.8, 2);

     \draw [->,shorten >=2pt,>=stealth] (8, 4) -- (9.8, 3);
     \draw [->,shorten >=2pt,>=stealth] (8, 4) -- (9.8, 2);

     \draw [->,shorten >=2pt,>=stealth] (12.1, 4) -- (13.8, 5);
     \draw [->,shorten >=2pt,>=stealth] (12.1, 4) -- (13.8, 1);

   \node at (-1.3,3) {$\Ser_1$};
   \draw [thick] (-0.8,3) -- (14,3);

   \draw [->,shorten >=2pt,>=stealth] (6.3, 3) -- (7.8, 4);
    \draw [->,shorten >=2pt,>=stealth] (10.3, 3) -- (11.9, 4);

    \node at (-1.3,2) {$\Ser_2$};
     \draw [thick] (-0.8,2) -- (14,2);

      \draw [->,shorten >=2pt,>=stealth] (6.3, 2) -- (7.9, 4);
     \draw [->,shorten >=2pt,>=stealth] (10.3, 2) -- (12.1, 4);

    \node at (-1.3,1) {$\Ser_3$};
    \node at (-1,0.6) {\scriptsize (passive)};
    \draw [thick] (-0.8,1) -- (14,1);

\end{tikzpicture}
        \caption{Message pattern in FastBFT.}
        \label{fig:FastBFT}
        \vspace{-0.8 em}
\end{figure}

\noindent\textbf{Hardware assistance.}
The strawman design is obviously insecure because $\Ser_p$, knowing the secret, can impersonate any $\Ser_i$. We fix this by making use of the TEE in each replica. The TEE in $\Ser_p$ generates secrets, splits them, and securely delivers shares to TEEs in each $\Ser_i$. During \commitP, the TEE of each $\Ser_i$ will release its share to $\Ser_i$ only if the prepare message is correct.
Notice that now $\Ser_p$ cannot reconstruct the secret without gathering enough shares from $\Ser_i$s.

Nevertheless, since secrets are generated during \preprocessingP, a faulty $\Ser_p$ can equivocate by using the same secret for different requests.
To remedy this, we have  $\Ser_p$'s TEE securely bind a secret to a counter value during \preprocessingP,
and during \prepareP, bind the request to the freshly incremented value of a TEE-resident monotonic counter.
This ensures that each specific secret is bound to a single request. TEEs of replicas keep track of $\Ser_p$'s latest counter value, updating their records after every successfully handled request. 
\revision{\textlabel{The key requirement here is that}{R1(3)} the TEE will neither use the same secret for different counter values nor use the same counter value for different secrets.}
To retrieve its share of a secret, $\Ser_i$ must present a 
prepare message with the right counter value to its local TEE.
%

In addition to maintaining and verifying monotonic counters like existing hardware-assisted BFT protocols \revision{(\textlabel{thus}{R2(8)}, it~requires $n = 2f + 1$ replicas to tolerate $f$ (Byzantine) faults)}, \BFT also uses TEEs for generating and sharing secrets.
 

\noindent {\bf Communication topology.} 
Even though this approach considerably reduces message complexity, $\Ser_p$ still needs to receive and aggregate $O(n)$ shares, which can be a bottleneck. To address this, we have $\Ser_p$ organize \revision{active} $\Ser_i$s into a balanced tree rooted at itself to distribute both communication and computation costs.
Shares are propagated along the tree in a bottom-up fashion: each intermediate node aggregates its children's shares together with its own; finally, $\Ser_p$ only needs to receive and aggregate \mbox{a small constant number of shares}.

\noindent {\bf Failure detection.}
Finally, \BFT adapts a failure detection mechanism from \cite{BChain} to tolerate non-primary faults.
Notice that a faulty node may simply crash or send a wrong share. 
A parent node is allowed to flag its direct children (and only them) as potentially faulty, and sends a suspect message up the tree.
Upon receiving this message, $\Ser_p$ replaces the accused replica with a passive replica and puts the accuser in a leaf so that it cannot continue to accuse others.





\begin{table}[tb]
\small
\centering
\begin{tabular}{|c|c|}
\hline
\textbf{Notation} & \textbf{Description} \\ \hline


$\Cli$       & Client                      \\ \hline
$\Ser$      & Replica                     \\ \hline



$n$ & Number of replicas \\ \hline
$f$ &  Number of faulty replicas \\ \hline
$p$ & Primary number  \\ \hline
$v$  & View number  \\ \hline
$c$ & Virtual counter value \\ \hline
$C$ & Hardware counter value \\ \hline


$H()$               &    Cryptographic hash function \\ \hline
$h$               &     Cryptographic hash             \\ \hline
$\E() / \D()$	& Authenticated encryption/decryption \\ \hline
$k$		& Key of authenticated encryption  \\ \hline
$\cipher$           & Ciphertext of authenticated encryption  \\ \hline
$\Enc()/\Dec()$	& Public-key encryption/decryption \\ \hline
$\omega$  & Ciphertext of public-key encryption  \\ \hline
$\Sign() / \Verify()$  & Signature generation / verification\\ \hline
$\langle x \rangle_{\sigma_i}$     & A Signature on $x$ by $\Ser_i$ \\ \hline

\end{tabular}
\caption{Summary of notations}
\label{notationtable}
\vspace{-2 em}
\end{table}

\section{\BFT: Detailed Design}
\label{sec:BFT}

In this section, we provide a full description of \BFT.
We introduce notations as needed (summarized in Table~\ref{notationtable}).

\subsection{TEE-hosted Functionality}
\label{subsec:TEE}

\begin{figure}[htbp]
\begin{mdframed}[leftline=false, rightline=false, linewidth=1pt, innerleftmargin=0cm, innerrightmargin=0cm]
\begin{algorithmic}[1]
\footnotesize
\State {\bf persistent variables:}
\State \phantom{----}{\bf maintained by all replicas:}
\State\label{1.01} \phantom{----}\phantom{----}$(c_{\textit{latest}}, v)$ \Comment{latest counter value and current view number}
\State \phantom{----}{\bf maintained by primary only:}
\State\label{1.03} \phantom{----}\phantom{----}$\{\Ser_i, k_i\}$ \Comment{current active replicas and their view keys}
\State\label{1.02} \phantom{----}\phantom{----}$T$ \Comment{current tree structure}
\State \phantom{----}{\bf maintained by active replica $\Ser_i$ only:}
\State\label{1.04} \phantom{----}\phantom{----}{\bf} $k_i$   \Comment{current view key agreed with the primary}
\Function{{\em be\_primary}($\{\Ser'_i\}, T'$)}{} \Comment{set $\Ser_i$ as the primary}
\State\label{1.53} $\{\Ser_i\}:=\{\Ser'_i\}$\phantom{----} $T := T'$\phantom{----} $v:=v+1$\phantom{----} \changed{$c := 0$}
\State \textbf{for} {each $\Ser_i$ in $\{\Ser_i\}$}
\State \phantom{----} $k_i \overset{\$}{\leftarrow} \{0,1\}^l$ \Comment{generate a random view key for $\Ser_i$}
\State\label{1.57} \phantom{----} $\omega_i \leftarrow \Enc(k_i)$ \Comment{encrypt $k_i$ using $\Ser_i$'s public key}
\State \Return $\{\omega_i\}$
\EndFunction
\\
\Function{{\em update\_view}($\langle x, (c, v) \rangle_{\sigma_{p'}}, \omega_i $)}{}\Comment{used by $\Ser_i$}
\State {\bf if} {$\Verify(\langle x, (c, v) \rangle_{\sigma_{p'}}) = 0$} \Return ``{\em invalid signature}''
\State {\bf else if} $c\neq c_{\textit{latest}}+1$ \Return ``{\em invalid counter}''
\State {\bf else}\label{1.67} \changed{$c_{\textit{latest}} := 0$}\phantom{----} $v := v +1$
\State\label{1.68} {\bf if} $\Ser_i$ is active, $k_i \leftarrow \Dec(\omega_i)$
\EndFunction
\\
\Function{{\em preprocessing}($m$)}{}\Comment{used by $\Ser_p$}
\State\textbf{for}  {$1 \leq a \leq m$}
\State\label{1.1} \phantom{----} $c := c_{\textit{latest}}+a$\phantom{----} $s_c \overset{\$}{\leftarrow} \{0,1\}^l$\phantom{----} $h_c \leftarrow H(\langle s_c, (c, v) \rangle)$ 
\State\label{1.3} \phantom{----} $s_c^1 \oplus ... \oplus s_c^{f+1} \leftarrow s_c$   \Comment{randomly splits $s_c$ into shares}
\State  \phantom{----}\textbf{for} {each active replica $\Ser_i$}
\State  \phantom{----} \phantom{----} \textbf{for}  {each of $\Ser_i$'s direct children: $\Ser_j$}
\State\label{1.4} \phantom{----} \phantom{----} \phantom{----}$\hat{h}_c^j:= H(s_c^j\oplus_{k\in \phi_j}s_c^k)$ \Comment{$\phi_j$ are $\Ser_j$'s descendants}
\State\label{1.20} \phantom{----} \phantom{----} $\cipher_c^i \leftarrow \E(k_i, \langle s_c^i, ( c,v), \{\hat{h}_c^j\}, h_c\rangle)$
\State\label{1.22} \phantom{----}$\langle h_c, ( c, v)\rangle_{\sigma_{p}} \leftarrow \Sign(\langle h_c, ( c, v)\rangle)$
\State \Return $\{\langle h_c, ( c, v) \rangle_{\sigma_{p}}, \{\cipher_c^i\}_i\}_c$
\EndFunction
\\
\Function{{\em request\_counter}($x$)}{}\Comment{used by $\Ser_p$}
\State\label{1.29} $c_{\textit{latest}} := c_{\textit{latest}}+1$ 
\State $\langle x, (c_{\textit{latest}}, v)\rangle_{\sigma} \leftarrow \Sign(\langle x, (c_{\textit{latest}}, v)\rangle)$
\State \Return $\langle x, (c_{\textit{latest}}, v)\rangle_{\sigma}$
\EndFunction
\\
\Function{{\em verify\_counter}($\langle x, ( c', v') \rangle_{\sigma_{p}}, \cipher_c^i$)}{}\Comment{used by active $\Ser_i$}
\State\label{1.34} {\bf if} {$\Verify(\langle x, ( c', v') \rangle_{\sigma_{p}}) = 0$} \Return ``{\em invalid signature}''
\State\label{1.35} {\bf else if} {$ \langle s_c^i, (c'',v''), \{\hat{h}_c^j\}, h_c\rangle \leftarrow \D(\cipher_c^i)$ fail} \Return ``{\em invalid encription}''
\State\label{1.37} {\bf else if} $(c', v') \neq (c'', v'')$ \Return ``{\em invalid counter value}''
\State\label{1.38} {\bf else if} $c' \neq c_{\textit{latest}} + 1$ \Return ``{\em invalid counter value}''
\State\label{1.40} {\bf else} $c_{\textit{latest}} := c_{\textit{latest}} + 1$ and \Return $\langle s_c^i, \{\hat{h}_c^j\}, h_c\rangle$
\EndFunction
\\
\Function{{\em update\_counter}($s_c, \langle h_{c}, (c, v) \rangle_{\sigma_p}$)}{}\Comment{by passive $\Ser_i$}
\State\label{1.45} {\bf if} {$\Verify(\langle h_{c}, (c, v) \rangle_{\sigma_p}) = 0$} \Return ``{\em invalid signature}''
\State {\bf else if} $c\neq c_{\textit{latest}}+1$ \Return ``{\em invalid counter}''
\State\label{1.47} {\bf else if} $H(\langle s_c, (c, v)\rangle) \neq h_c$ \Return ``{\em invalid secret}''
\State\label{1.48} {\bf else} $c_{\textit{latest}} := c_{\textit{latest}} + 1$
\EndFunction
\\
\Function{{\em reset\_counter}($\{L_i, \langle H(L_{i}), (c', v') \rangle_{\sigma_i}\}$)}{}\Comment{by $\Ser_i$}
\State {\bf if} at least $f+1$ consistent $L_i$, $(c', v')$
\State\label{1.49} \phantom{----} $c_{\textit{latest}} := c'$ and $v := v'$
\EndFunction
\end{algorithmic}
\end{mdframed}
\caption{TEE-hosted functionality required by \BFT.}
\label{fig:tee}
\end{figure}

\fig{fig:tee} shows the TEE-hosted functionality required by \BFT.
Each TEE is equipped with certified keypairs to encrypt data for that TEE (using \Enc()) and to generate signatures (using \Sign()).
The primary $\Ser_p$'s TEE maintains a monotonic counter with value $c_{\textit{latest}}$; TEEs of other replicas $\Ser_i$s keep track of $c_{\textit{latest}}$ and the current view number $v$ (line~\ref{1.01}).
$\Ser_p$'s TEE also keeps track of each currently active $\Ser_i$, key $k_i$ shared with $\Ser_i$ (line~\ref{1.03}) and the tree topology $T$ for $\Ser_i$s (line~\ref{1.02}).
Active $\Ser_i$s also keep track of their $k_i$s (line~\ref{1.04}). Next, we describe each TEE function.

\noindent{\bf \em be\_primary}:
asserts a replica as primary by
setting $T$, incrementing $v$, \changed{re-initializing $c$ (line~\ref{1.53})},
and generating $k_i$ for each active $\Ser_i$'s TEE (line~\ref{1.57}).

\noindent{\bf \em update\_view}: enables all replicas to update $(c_{\textit{latest}}, v)$ (line~\ref{1.67}) and new active replicas to receive and set $k_i$ from $\Ser_p$ (line~\ref{1.68}).

\noindent{\bf \em preprocessing}:
for each preprocessed counter value $c$,
generates a secret $s_c$ together with its hash $h_c$ (line~\ref{1.1}),
$f+1$ shares of $s_c$ (line~\ref{1.3}),
and $\{\hat{h}_c^j\}$ (line~\ref{1.4}) that allows each $\Ser_i$ to verify its children's shares.
Encrypts these using authenticated encryption with each $k_i$ (line~\ref{1.20}).
Generates a signature $\sigma_{p'}$ (line~\ref{1.22}) to bind $s_c$ with \mbox{the counter value $(c,v)$}.


\noindent{\bf \em request\_counter}: increments $c_{\textit{latest}}$ and binds it (and $v$) to
the input $x$ by signing them (line~\ref{1.29}).

\noindent{\bf \em verify\_counter}:
receives $\langle h, ( c', v') \rangle_{\sigma_{p}}, \cipher_c^i$; verifies:
(1) validity of $\sigma_{p}$ (line~\ref{1.34}), (2) integrity of $\cipher_c^i$ (line~\ref{1.35}),
(3) whether the counter value and view number inside $\cipher_c^i$ match $(c', v')$ (line~\ref{1.37}),
and (4) whether $c'$ is equal to $c_{\textit{latest}}+1$ (line~\ref{1.38}).
Increments $c_{\textit{latest}}$ and returns $\langle s_c^i, \{\hat{h}_c^j\}, h_c\rangle$ (line~\ref{1.40}).

\noindent{\bf \em update\_counter}:
receives $s_c, \langle h_c, (c,v)\rangle_{\sigma_p}$; verifies $\sigma_{p}$, $c$ and $s_c$ (line~\ref{1.45}-\ref{1.47}).
Increments $c_{\textit{latest}}$ (line~\ref{1.48}).

\noindent{\bf \em reset\_counter}:
\revision{
\textlabel{receives at least}{R1:9} {\em (f+1)} ($L_i, (c', v')$)s; 
sets $c_{\textit{latest}}$ as $c'$ and $v$ as $v'$ (line~\ref{1.49}).
}


\subsection{Normal-case Operation}
\label{sec:normal_case}

\begin{figure}[p]
\small
\begin{mdframed}[leftline=false, rightline=false, linewidth=1pt, innerleftmargin=0cm, innerrightmargin=0cm]
\begin{algorithmic}[1]
\State \textbf{upon} invocation of PREPROCESSING at $\Ser_p$ \textbf{do}
\State\label{2} \phantom{----} $\{\langle h_c, ( c, v) \rangle_{\sigma_{p}}, \{\cipher_c^i\}_i\}_c$ $\leftarrow$ $\TA$.{\em preprocessing}($m$)
\State\label{3} \phantom{----} \textbf{for} each active $\Ser_i$, send $\{\cipher_c^i\}_c$ to $\Ser_i$
\\
\State \textbf{upon} reception of $M = \langle \request, op \rangle_{\sigma_{\Cli}}$ at $\Ser_p$ \textbf{do}
\State\label{7} \phantom{----} $\langle H(M), ( c, v) \rangle_{\sigma_{p}}$ $\leftarrow$ $\TA$.{\em request\_counter}($H(M)$)
\State\label{8} \phantom{----} multicast $\langle \prepare, M, \langle H(M), ( c, v) \rangle_{\sigma_{p}}\rangle$ to active $\Ser_i$s
\\
\State \textbf{upon} reception of $\langle \prepare, M, \langle H(M), ( c, v) \rangle_{\sigma_{p}}\rangle$ at $\Ser_i$ \textbf{do}
\State\label{11} \phantom{----} $\langle s_c^i,\{\hat{h}_c^j\}, h_c\rangle$ $\leftarrow$ $\TA$.{\em verify\_counter}($\langle H(M), (c, v) \rangle_{\sigma_{p}}$, $\cipher_c^i$)
\State\label{12} \phantom{----} $\hat{s}_c^i := s_c^i$
\State\label{13} \phantom{----} \textbf{if} $\Ser_i$ is a leaf node, send $s_c^i$ to its parent
\State\label{14} \phantom{----} \textbf{else}  set timers for its direct children
\\
\State \textbf{upon} timeout of $\Ser_j$'s share at $\Ser_i$ \textbf{do}
\State\label{15} \phantom{----} send $\langle \suspect, \Ser_j \rangle$ to both $\Ser_p$ and $\Ser_j$'s parent
\\
\State \textbf{upon} reception of $\hat{s}_c^j$ at $\Ser_i$/$\Ser_p$ \textbf{do}
\State\label{27} \phantom{----}  \textbf{if} $H(\hat{s}_c^j) = \hat{h}_c^j$, $\hat{s}_c^i := \hat{s}_c^i \oplus \hat{s}_c^j$
\State\label{28} \phantom{----}  \textbf{else} send $\langle \suspect, \Ser_j \rangle$ $\Ser_p$  
\State\label{29} \phantom{----}\phantom{----}  \textbf{if} $i\neq p$, send to its parent
\State\label{30} \phantom{----}  \textbf{if} $\Ser_i$ has received all valid $\{\hat{s}_c^j\}_j$, send $\hat{s}_c^i$ to its parent
\State\label{31} \phantom{----}  \textbf{if} $\Ser_p$ has received all valid $\{\hat{s}_c^j\}_j$
\State\label{32} \phantom{----}\phantom{----} \textbf{if} \revision{$s_c$ \textlabel{is used for the commit phase}{R1(18.1)}}
\State\label{44} \phantom{----}\phantom{----}\phantom{----}$res$ $\leftarrow$ execute $op$ \phantom{--} $x \leftarrow H(M||res)$
\State\label{45} \phantom{----}\phantom{----}\phantom{---}$\langle x, (c+1, v) \rangle_{\sigma_{p}}$ $\leftarrow$ $\TA$.{\em request\_counter}($x$)
\State\label{46} \phantom{----}\phantom{----}\phantom{---} send active $\Ser_i$s $\langle \commit, s_c, res, \langle x, (c+1, v) \rangle_{\sigma_{p}}\rangle$
\State \phantom{----}\phantom{----}  \textbf{else if} \revision{$s_c$ is used for the reply phase}
\State\label{48} \phantom{----}\phantom{----}\phantom{---} send 
\revision{$\langle \reply, \textlabel{M, res, s_{c-1}, s_{c}}{r1:1},$} 
\revision{$\langle h_{c-1}, (c-1,v) \rangle_{\sigma_{p}},$} 
\revision{$\langle h_{c}, (c,v) \rangle_{\sigma_{p}},$}. 
$\langle H(M), $ $ (c-1,v) \rangle_{\sigma_{p}},$ 
$\langle H(M||res), (c, v)\rangle_{\sigma_{p}}\rangle$ 
to $\Cli$ and passive replicas.
\\
\State \textbf{upon} reception of $\langle \suspect, \Ser_k \rangle$ from $\Ser_j$ at $\Ser_i$ \textbf{do}
\State\label{20}  \phantom{----}  \textbf{if} $i = p$
\State\label{21}  \phantom{----}\phantom{--} generate new tree $T'$ replacing $S_k$ with a passive replica and placing $S_j$ at a leaf.
\State\label{22}  \phantom{----}\phantom{--} $\langle H(T||T'), (c, v)\rangle_{\sigma_{p}} \rangle$ $\leftarrow$ TEE.{\em request\_counter}($H(T||T')$)
\State\label{23}  \phantom{----}\phantom{--} broadcast $\langle \newtree, T, T', \langle H(T||T'), (c, v)\rangle_{\sigma_{p}} \rangle$
\State\label{24}  \phantom{---} \textbf{else} cancel $\Ser_j$'s timer and forward the $\suspect$ message up
\\
\State \textbf{upon} reception of $\langle \commit, s_{c}, res, \langle H(M||res), (c+1, v) \rangle_{\sigma_{p}}\rangle$ at $\Ser_i$ \textbf{do}
\State\label{51} \phantom{----} {\bf if} $H(s_{c}) \neq h_{c}$ {\bf or} execute $op$ $\neq res$
\State\label{52} \phantom{----}\phantom{--} broadcast $\langle \reqviewchange,  v, v' \rangle$
\State\label{53} \phantom{----}\phantom{--}  
$\langle s_{c+1}^i, \{\hat{h}_{c+1}^j\}, h_{c+1}\rangle$ 
$\leftarrow$ 
$\TA$.{\em verify\_counter} 
($\langle H(M||res), $ 
$(c+1,v) \rangle_{\sigma_{p}}$, 
$\cipher_c^i$)
 
\State\label{55} \phantom{----}  \textbf{if} $\Ser_i$ is a leaf node, send $s_{c+1}^i$ to its parent
\State\label{56} \phantom{----} \textbf{else} $\hat{s}_{c+1}^i := s_{c+1}^i$, set timers for its direct children
\\
\State \textbf{upon} \textlabel{reception of}{r1:2} \revision{$\langle \reply, M, res, s_{c}, s_{c+1}, \langle h_{c}, (c,v) \rangle_{\sigma_{p}},  \langle h_{c+1},$ $(c+1,v) \rangle_{\sigma_{p}},$}
$\langle H(M), (c,v) \rangle_{\sigma_{p}}, \langle H(M||res), (c+1, v)\rangle_{\sigma_{p}}\rangle$ at $\Ser_i$ \textbf{do}
\State\label{59} \phantom{----} {\bf if} $H(s_c) \neq h_c$ {\bf or} $H(s_{c+1}) \neq h_{c+1} $
\State \phantom{----}\phantom{----}  multicasts $\langle \reqviewchange,  v, v' \rangle$
\State\label{60} \phantom{----} {\bf else} update state based on $res$
\State\label{61} \phantom{----}\phantom{---} TEE.{\em update\_counter}($s_{c},\langle h_{c}, (c,v) \rangle_{\sigma_{p}}$)
\State\label{62} \phantom{----}\phantom{---} TEE.{\em update\_counter}($s_{c+1},\langle h_{c+1}, (c+1,v) \rangle_{\sigma_{p}}$)

\end{algorithmic}
\end{mdframed}
\caption{Pseudocode: normal-case operation with failure detection.}
\label{fig:normal_case}
\end{figure}


Now we describe the normal-case operation of a replica as a reactive system (\fig{fig:normal_case}). For the sake of brevity, we do not explicitly show signature verifications and we assume that each replica verifies any signature received as input.

\noindent\textbf{Preprocessing.}
$\Ser_p$ decides the number of preprocessed counter values (say $m$),
and invokes {\em preprocessing} on its TEE (line~\ref{2}).
$\Ser_p$ then sends the resulting package $\{\cipher_c^i\}_c$ to each $\Ser_i$ (line~\ref{3}).

\noindent\textbf{Request.}
A client $\Cli$~requests execution of $op$ by sending a signed request $M = \langle \request, op \rangle_{\sigma_{\Cli}}$ to $\Ser_p$. 
If $\Cli$ receives no reply before a timeout,  it broadcasts\footnote{\scriptsize{We use the term ``broadcast'' when a message is sent to all replicas, and ``multicast'' when it is sent to a subset of replicas.}} $M$.

\noindent\textbf{Prepare.}
Upon receiving $M$, $\Ser_p$
invokes {\em request\_counter} with $H(M)$ to get a signature binding $M$ to $(c,v)$ (line~\ref{7}).
$\Ser_p$ multicasts $\langle \prepare, M, \langle H(M), ( c, v) \rangle_{\sigma_{p}}\rangle$ to all active $\Ser_i$s (line~\ref{8}).
\ifsubmission
\else
This can be achieved either by sending the message along the tree or by using direct multicast, depending on the underlying topology.
\fi
At this point, the request $M$ is {\em prepared}.

\noindent\textbf{Commit.}
Upon receiving the $\prepare$ message,
each $\Ser_i$
invokes {\em verify\_counter} with $\langle H(M), (c, v) \rangle_{\sigma_{p}}$ and the corresponding $\cipher_c^i$, and receives $\langle s_c^i, \{\hat{h}_c^j\}, h_c\rangle$ as output (line~\ref{11}).

If $\Ser_i$ is a leaf node, it sends $s_c^i$ to its parent (line~\ref{13}).
Otherwise, $\Ser_i$ waits to receive a partial aggregate share $\hat{s}_c^j$ from each of its immediate children $\Ser_j$ and verifies if $H(\hat{s}_c^j) = \hat{h}_c^j$ (line~\ref{27}).
If this verification succeeds, $\Ser_i$ computes $\hat{s}_c^i = s_c^i \oplus_{j\in \phi_i} \hat{s}_c^j$ where $\phi_i$ is the set of $\Ser_i$'s children (line~\ref{30}).

Upon reconstructing the secret $s_c$, $\Ser_p$ executes $op$ to obtain $res$ (line~\ref{44}), and
multicasts $\langle \commit, s_c, res, $ $\langle H(M||res), $ $(c+1, v)\rangle_{\sigma_{p}}\rangle$ 
to all active $\Ser_i$s (line~\ref{46})\footnote{\scriptsize{In case the execution of $op$ takes long, $\Ser_p$ can multicast $s_c$ first and multicast the \commit message when execution completes.}}.
At this point, $M$ is {\em committed}.



\noindent\textbf{Reply.} 
Upon receiving the $\commit$ message, each active $\Ser_i$ verifies $s_c$ against $h_c$,
and executes $op$ to acquire the result $res$ (line~\ref{51}).
$\Ser_i$ then executes a procedure similar to \commitP to open $s_{c+1}$ (line~\ref{53}-\ref{56}).
$\Ser_p$ sends \revision{$\langle \reply, M, res, s_c, s_{c+1},\langle h_c, (c,v) \rangle_{\sigma_{p}}, \langle h_{c+1}, (c+1,v) \rangle_{\sigma_{p}},\langle H(M), (c,v) \rangle_{\sigma_{p}},\langle H(M||res),$  $(c+1, v)\rangle_{\sigma_{p}}\rangle$}
to $\Cli$ as well as to all passive replicas
(line~\ref{48}).
At this point $M$ has been {\em replied}.
$\Cli$ verifies the validity of this message:

\begin{enumerate}
\item A valid $\langle h_c, (c, v) \rangle_{\sigma_{p}}$ implies that $(c,v)$ was bound to a secret $s_c$ whose hash is $h_c$.
\revision{\textlabel{This implication holds only if $s_c$ is not reused}{R1(10)}, which is an invariant that our protocol ensures}
\item A valid $\langle H(M), (c,v) \rangle_{\sigma_{p}}$ implies that $(c,v)$ was bound to the request message $M$.
\item Thus, $M$ was bound to $s_c$ based on 1) and 2).
\item A valid $s_c$ (i.e., $H(s_c, (c, v)) = h_c$) implies that all active $\Ser_i$s have agreed to execute $op$ with counter value $c$.
\item A valid $s_{c+1}$ implies that all active $\Ser_i$s have executed $op$, which yields $res$.
\end{enumerate}
Each passive replica performs this verification, updates its state (line~\ref{60}),
and transfers the signed counter values to its local $\TA$ to update the latest counter value (line~\ref{61}-\ref{62}).

\revision{A communication structure for the $\commitP$/$\replyP$ phase is shown in Figure~\ref{fig:structure}.}

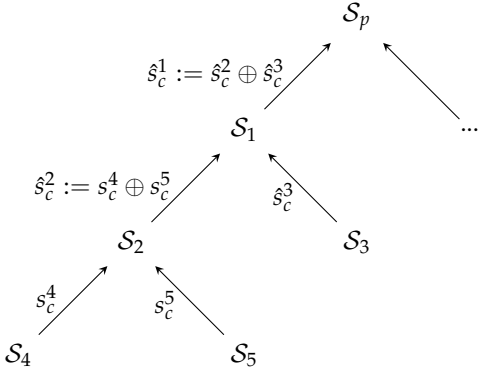
\begin{figure}[tb]
\begin{tikzpicture}
  [scale=.75,auto=center]
  \node (n0) at (8,8)  {$\Ser_p$};
  \node at (5.5,7)  {$\hat{s}^1_c:=\hat{s}^2_c\oplus \hat{s}^3_c$};
  \node (n1) at (6,6)  {$\Ser_1$};
  \node at (3.5,5)  {$\hat{s}^2_c:=s^4_c\oplus s^5_c$};
  \node (n2) at (4,4)  {$\Ser_2$};
   \node at (6.7,4.8)  {$\hat{s}^3_c$};
  \node (n4) at (8,4)  {$\Ser_3$};
   \node at (2.5,3)  {$s^4_c$};
  \node (n5) at (2,2)  {$\Ser_4$};
  \node (n7) at (6,2)  {$\Ser_5$};
  \node at (4.6,2.88)  {$s^5_c$};
  \node (n12) at(10,6)  {...};
  \foreach \from/\to in {n1/n2,n1/n4,n2/n5,n2/n7,n0/n1,n0/n12}
  	\draw  [->,shorten >=2pt,>=stealth]  (\to) --   (\from);
\end{tikzpicture}
        \caption{Communication structure for the $\commitP$/$\replyP$ phase.}
        \label{fig:structure}
\end{figure}

\subsection{Failure Detection}
\label{sec:failure_detection}

\changedagain{
Unlike classical BFT protocols which can tolerate non-primary faults for free,
optimistic BFT protocols usually require transitions~\cite{CheapBFT} or view-changes~\cite{XFT}.
To tolerate non-primary faults in a more efficient way, \BFT leverages an efficient failure detection mechanism.}

\changedagain{Similar to previous BFT protocols~\cite{PBFT, MinBFT}, we rely on timeouts to detect crash failures and we have parent nodes detect their children's failures by verifying shares.}
Specifically, upon receiving a $\prepare$ message, $\Ser_i$ starts a timer for each of its direct children (\fig{fig:normal_case}, line~\ref{14}).
If $\Ser_i$ fails to receive a share from $\Ser_j$ before the timer expires (line~\ref{15}) or if $\Ser_i$ receives a wrong share that does not match $\hat{h}_c^j$ (line~\ref{28}),
it sends $\langle \suspect, \Ser_j \rangle$ to its parent and $\Ser_p$ to signal potential failure of $\Ser_j$.
Whenever a replica receives a $\suspect$ message from its child,
it cancels the timer of this child to reduce the number of $\suspect$ messages,
and forwards this $\suspect$ message to its parent along the tree until it reaches the root $\Ser_p$ (line~\ref{24}).
For multiple $\suspect$ messages along the same path, $\Ser_p$ only handles the node that is closest \mbox{to the leaf}.

Upon receiving $\suspect$,
$\Ser_p$ broadcasts $\langle \newtree,$  $T, T', $ $\langle H(T||T'), (c, v)\rangle_{\sigma_{p}} \rangle$ (line~\ref{23}),
where $T$ is the old tree and $T'$ the new tree.
$\Ser_p$ replaces the accused replica $\Ser_j$ with a randomly chosen passive replica
and moves the accuser $\Ser_i$ to a leaf position to prevent the impact of a faulty accuser continuing to incorrectly report other replicas as faulty.
\textlabel{Notice that this allows a Byzantine}{R2(10)} $\Ser_p$ to evict correct replicas.
However, there will always be at least one correct replica among the $f+1$ active replicas. 
\changedagain{Notice that $\Ser_j$ might be replaced by a passive replica if it did not receive a \prepare/\commit message and thus failed to provide a correct share.
In this case, its local counter value will be smaller than that of other correct replicas.
To rejoin the protocol, $\Ser_j$ can ask $\Ser_p$ for the PREPARE/COMMIT messages to update its counter.}

If there are multiple faulty nodes along the same path, the above approach can only detect one of them within one round.
We can extend this approach by having $\Ser_p$ check correctness of all active replicas individually after one failure detection to allow detection of multiple \mbox{failures within one round.}

\changedagain{\textlabel{Notice that $f$ faulty replicas}{R1(2)} can take advantage of the failure detection mechanism to trigger a sequence of tree reconstructions (i.e., cause a denial of service DoS attack). After the number of detected non-primary failures exceed a threshold, $\Ser_p$ can trigger a transition protocol~\cite{CheapBFT} to fall back to a classical BFT protocol (cf. Section~\ref{sec:minBFT}).}

\subsection{View-change}
\label{sec:view-change}

\changed{Recall that $\Cli$ sets a timer after sending a request to $\Ser_p$. It will broadcast the request \revision{to all replicas} if no reply was received before the timeout.
If a replica receives no \prepare (or \commit/\reply) message before the timeout,
it will initialize a view-change (\fig{fig:view_change}) by broadcasting a $\langle \reqviewchange, L, \langle H(L), $ $(c, v) \rangle_{\sigma_{i}} \rangle$ message,
where $L$ is the message log  that includes all 
messages it has received/sent since the latest checkpoint\footnote{\scriptsize{Similar to other BFT protocols, \BFT generates checkpoints periodically to limit the number of messages in the log.}}.}
In addition, replicas can also suspect that $\Ser_p$ is faulty by verifying the messages they received and initialize a view-change (i.e., line~\ref{11}, line~\ref{51},~\ref{59} in \fig{fig:normal_case}).
\revision{\textlabel{Notice that passive replicas}{R1(12)} can also send $\reqviewchange$ messages. Thus, if faulty primary occurs, there will be always $f+1$ non-faulty replicas initiate the view-change.}

\changedagain{
Upon receiving $f+1$ $\reqviewchange$ messages, the new primary $\Ser_{p'}$ (that satisfies $p' = v'$ mod $n$)
constructs the execution history $O$ by collecting all prepared/committed/replied requests from the message logs (line~\ref{2.6}).
Notice that there might be an existing valid execution history in the message logs due to previously failed view-changes.
In this case, $\Ser_{p'}$ just uses that history.
This strategy guarantees that replicas will always process the same execution history.}
$\Ser_{p'}$ also constructs a tree $T'$ that specifies $f+1$ new active replicas for view $v'$ (line~\ref{2.7}).
Then, it invokes $be\_primary$ on its TEE to record $T'$ and generate a set of shared view keys for the new active replicas' TEEs (line~\ref{2.8}).
Next, $\Ser_{p'}$ broadcasts $\langle \newview, O, T', \langle H({O}||T'), (c+1, v) \rangle_{\sigma_{p'}}, \{\omega_i\} \rangle$ (line~\ref{2.10}).

\changed{
Upon receiving a $\newview$ message from $\Ser_{p'}$, $\Ser_i$ verifies whether ${O}$ was constructed properly,
and broadcasts $\langle \viewchange,$ $ \langle H(O||T'), (c+1,v) \rangle_{\sigma_{i}} \rangle$ (line~\ref{2.13.2}).}
Upon receiving $f$ \viewchange messages\footnote{\scriptsize{$\Ser_{p'}$ uses \newview to represent its \viewchange message, so it is actually $f+1$ \viewchange messages.}},
$\Ser_i$ executes all requests in ${O}$ that have not yet been executed locally, following the counter values (line~\ref{2.13.3}).
\changedagain{A valid \newview message and $f$ valid \viewchange messages represent that $f+1$ replicas have committed to execute the requests in $O$.}
After execution, $\Ser_i$ begins the new view by invoking $update\_view$ on its local TEE (line~\ref{2.14}).


The new set of active replicas run the preprocessing phase for view $v'$, reply to the requests that have not been yet replied, and process the requests that have \mbox{not yet been prepared.}

\revision{\textlabel{The view-change protocol potentially leads to}{R1:4}  counters out of sync.
Suppose there is a quorum $Q$ of less than $f+1$ replicas receive no message after a \prepare message with a counter value $(c, v)$, 
they will keep sending a \reqviewchange with a counter value $(c+1, v)$. 
On the other hand, there is a quorum $Q'$ of at least $f+1$ replicas are still in the normal-operation and keep increasing their counters, 
$(c+1, v), (c+2, v), ..., (c+x, v)$. 
In this case, the replicas in $Q$ cannot rejoin $Q'$ because their counter values are out of sync, 
but the safety and liveness are still hold as long as the replicas in $Q'$ follow the protocol. 
Next, consider some replicas in $Q'$ misbehave and other replicas initiate a \viewchange by sending  
\reqviewchange with $(c+x+1, v)$. 
Now, there will be more than $f+1$ \reqviewchange messages and the view-change will happen. 
The honest replicas in $Q$ will execute the operations up to $(c+x+1, v)$ based on the execution history sent by the replicas in $Q'$. 
Then, all replicas will switch to a new view with a new counter value $(0, v+1)$. 
}

\begin{figure}[htbp]
\begin{mdframed}[leftline=false, rightline=false, linewidth=1pt, innerleftmargin=0cm, innerrightmargin=0cm]
\begin{algorithmic}[1]
\small
\State \textbf{upon} reception of $f+1$ $\langle \reqviewchange, L, $ $\langle H(L), (c, v) \rangle_{\sigma_{i}} \rangle$ messages at the new primary $\Ser_p'$ \textbf{do}
\State\label{2.6} \phantom{--} build execution history $O$ based on message logs $\{L\}$
\State\label{2.7} \phantom{--} choose $f+1$ new active replicas and construct a tree $T'$
\State\label{2.9} \phantom{-} $\langle H(O||T'), ( c+1, v) \rangle_{\sigma_{p'}}$ $\leftarrow$ $\TA$.{\em request\_counter}($H(O||T')$)
\State\label{2.8} \phantom{--} $\{\omega_i\}$ $\leftarrow$ $\TA$.{\em be\_primary}($\{\Ser_i\}, T'$)
\State\label{2.10} \phantom{-} broadcast $\langle \newview, O, T', \langle H(O||T'), $ 
$(c+1, v) \rangle_{\sigma_{p'}}, $ 
$\{\omega_i\} \rangle$
\\
\State \textbf{upon} reception of $\langle \newview, O, T', \langle H(O||T'), (c+1, v) \rangle_{\sigma_{p'}},$ $\{\omega_i\} \rangle$ at $\Ser_{i}$ \textbf{do}
\State\label{2.13} \phantom{----} {\bf if} $O$ is valid
\State\label{2.13.1} \phantom{---}\phantom{--} $\langle H(O||T'), (c+1, v) \rangle_{\sigma_{i}}$ $\leftarrow$ $\TA$.{\em request\_counter}
$($ $H(O||T'))$
\State\label{2.13.2} \phantom{----}\phantom{----} broadcast $\langle \viewchange, \langle H(O||T'), (c+1,v) \rangle_{\sigma_{i}} \rangle$

\\
\State \textbf{upon} reception of $f$  $\langle \viewchange, \langle H(O||T'), (c+1,v) \rangle_{\sigma_{i}} \rangle$ messages at $\Ser_{i}$ \textbf{do}
\State\label{2.13.3} \phantom{----} execute the requests in $O$ that have not been executed 
\State\label{2.13.4} \phantom{----} extract and store information from $T'$
\State\label{2.14} \phantom{----} $\TA$.{\em update\_view}($\langle H(O||T'), (c+1, v) \rangle_{\sigma_{p'}}, \omega_i \rangle$)
\end{algorithmic}
\end{mdframed}
\caption{Pseudocode: view-change.}
\label{fig:view_change}
\vspace{-0.18em}
\end{figure}

\subsection{\changedagain{Fallback Protocol: classical BFT with message aggregation}}
\label{sec:minBFT}

As we mentioned in Section~\ref{sec:failure_detection}, after a threshold number of failure detections, $\Ser_p$ initiates a {\em transition protocol}, which is exactly the same as the view-change protocol in Section~\ref{sec:view-change}, 
to reach a consensus on the current state and switch to the next ``view'' without changing the primary.
Next, all replicas run the following classical BFT as fallback instead of running the normal-case operation.
Given that permanent faults are rare, \BFT stays in this fallback mode for a fixed duration 
after which it will attempt to transition back to normal-case.
\revision{\textlabel{Before switching back to normal-case operation}{R1(11)}, $\Ser_p$ check replicas' states by broadcasting a message and asking for responses. In this way, $\Ser_p$ can avoid choosing crashed replicas to be active.
Then, $\Ser_p$ initiates a protocol that is similar to view-change but set itself as the primary.
If all $f+1$ potential active replicas participate in the view change protocol, they will successfully switch back to the normal-case operation.
}

\changedagain{To this end, we propose a new classical BFT protocol which combines the use of MinBFT with our hardware-assisted message aggregation technique.}
Unlike speculative or optimistic BFT where all (active) replicas are required to commit and/or reply, 
classical BFT only requires a subset (e.g., $f+1$ out of $2f+1$) replicas to commit and reply.
When applying our techniques to classical BFT, one needs to use a ($f+1$)-out-of-$(2f+1)$ secret sharing technique,
such as Shamir's polynomial-based secret sharing, rather than the XOR-based secret sharing.
In MinBFT, $\Ser_p$ broadcasts a $\prepare$ message including a monotonic counter value.
Then, each $\Ser_i$ broadcasts a $\commit$ message to others to agree on the proposal from $\Ser_p$.
To get rid of all-to-all multicast, we again introduce a preprocessing phase,
where $\Ser_p$'s local $\TA$ first generates $n$ random shares $x_1, ..., x_n$, and for each $x_i$, computes $\{\frac{x_j}{x_j-x_i}\}_j$ together with  $(x_i^2, ..., x_i^f)$. 
Then, for each counter value $c$, $\Ser_p$ performs the following operations:
\begin{enumerate}

\item $\Ser_p$ generates a polynomial with independent random coefficients: $f_c(x) = s_c + a_{1,c}x^1 + ... + a_{f,c}x^f$ where $s_c$ is a secret to be shared.

\item $\Ser_p$ calculates $h_c \leftarrow H(s_c, (c,v))$.

\item For each active $\Ser_i$, $\Ser_p$ calculates $\cipher_c^i$ = $E(k_i, \langle (x_i, f_c(x_i)), $ $(c,v), h_c\rangle)$.

\item $\Ser_p$ invokes its TEE to compute $\langle h_c, (c, v)\rangle_{\sigma_{p}}$ which is a signature generated using the signing key inside $\TA$.

\item $\Ser_p$ gives $\langle h_c, (c, v)\rangle_{\sigma_{p}}$ and $\{\cipher_c^i\}$ to $\Ser_p$.
\end{enumerate}
Subsequently, $\Ser_p$ sends $\cipher_c^i$ to each replica $\Ser_i$.
Later, in the commit phase, after receiving at least $f+1$ shares, $\Ser_p$ reconstructs the secret: $s_c = \sum_{i=1}^{f+1}(f_c(x_i)\prod_{j\neq i}\frac{x_j}{x_j-x_i})$.
With this technique, the message complexity of MinBFT is reduced from $O(n^2)$ to $O(n)$. However, the polynomial-based secret sharing is more expensive than the XOR-based one used in \BFT.

\changedagain{The fallback protocol does not rely on the tree structure since a faulty node in the tree can make its whole subtree ``faulty''---thus the fallback protocol can no longer tolerate non-primary faults for free.
If on the other hand primary failure happens in the fallback protocol, replicas execute the same view-change protocol as normal-case.}

\section{Correctness of \BFT}
\label{sec:correctness}


\revision{In this section, \textlabel{we provide an informal argument}{r1:5.4} for the correctness of \BFT. A formal (ideally machine-checked) proof of safety and liveness is left as future work.}

\subsection{Safety}

We show that if a correct replica executed a sequence of operations $\langle op_1, ..., op_m\rangle$, then all other correct replicas executed the same sequence of operations or a prefix of it. 
\begin{lemma}
\label{safety_same_view}
{\em In a view $v$, if a correct replica executes an operation $op$ with counter value $(c,v)$, no correct replica executes a different operation $op'$ with this counter value.}
\end{lemma}
\begin{proof}
Assume two correct replicas $\Ser_i$ and $\Ser_j$ executed two different operations $op_i$ and $op_j$ with the same counter value $(c,v)$.
There are following cases:
\begin{enumerate}
\item {\em Both $\Ser_i$ and $\Ser_j$ executed $op_i$ and $op_j$ during normal-case operation.}
In this case, they must have received valid \commit (or \reply) messages with $\langle H(M_i||res_i), (c, v)\rangle_{\sigma_{p}}$ and $\langle H(M_j||res_j), (c, v)\rangle_{\sigma_{p}}$ respectively
(\fig{fig:normal_case}, line~\ref{46} and line~\ref{48}).
This is impossible since $\Ser_p$'s $\TA$ will never sign different requests 
 with \mbox{the same counter value.}
\item {\em $\Ser_i$ executed $op_i$ during normal-case operation while $\Ser_j$ executed $op_j$ during view-change operation.}
In this case, $\Ser_i$ must have received a $\commit$ (or $\reply$) message for $op_i$ with an ``opened'' secret $s_{c-1}$.
To open $s_{c-1}$, a quorum $Q$ of $f +1$ active replicas must provide their shares (\fig{fig:normal_case}, line~\ref{31}).
This also implies that they have received a valid $\prepare$ message for $op_i$ with $(c-1, v)$ and their TEE-recorded counter value is at least $c-1$ (\fig{fig:normal_case}, line~\ref{11}).
Recall that before changing to the next view, $\Ser_j$ will process an execution history $O$ based on message logs provided by a quorum $Q'$ of at least $f+1$ replicas (Figure~\ref{fig:view_change}, line~\ref{2.6}).
\revision{\textlabel{So, there must be an intersection replica}{R1(18)} $\Ser_k$ between $Q$ and $Q'$, which includes the $\prepare$ message for $op_i$ in its message log,
otherwise the counter values will not be sequential.}
Therefore, a correct $\Ser_j$ will execute the operation $op_i$ with counter value $(c, v)$ before changing to the next view (\fig{fig:view_change}, line~\ref{2.13.3}).
\item {\em Both $\Ser_i$ and $\Ser_j$ execute $op_i$ and $op_j$ during view-change operation.}
They must have processed the execution histories that contains the \prepare messages for $op_i$ and $op_j$ respectively.
$\Ser_p$'s $\TA$ guarantees that $\Ser_p$ cannot generate different \prepare messages with the same counter value.
 \changedagain{
 \item {\em Both $\Ser_i$ and $\Ser_j$ execute $op_i$ and $op_j$ during the fallback protocol.}
Similar to case 1, they must have received valid \commit messages with $\langle H(M_i||res_i), (c, v)\rangle_{\sigma_{p}}$
and $\langle H(M_j||res_j),$ $(c, v)\rangle_{\sigma_{p}}$ respectively, which \mbox{is impossible}.
\item {\em $\Ser_i$ executed $op_i$ during the fallback protocol while $\Ser_j$ executed $op_j$ during view-change operation. }
The argument for this case is the same as case 2.
}
\end{enumerate}

Therefore, we conclude that it is impossible for two different operations to be executed with the same counter value \mbox{during a view}.
\end{proof}

\begin{lemma}
\label{safety_different_view}
{\em If a correct replica executes an operation $op$ in a view $v$, no correct replica will change to a new view without \mbox{executing $op$}.}
\end{lemma}

\begin{proof}
Assume that a correct replica $\Ser_i$ executed $op$ in view $v$, and another correct replica $\Ser_j$ change to the next view without executing $op$.
We distinguish between two cases:

\begin{enumerate}
\item {\em $\Ser_i$ executed $op$ during normal-case operation \changedagain{(or during fallback)}.}
As mentioned in Case 2 of the proof of Lemma~\ref{safety_same_view}, the \prepare message for $op$ will be included in the execution history $O$.
Therefore, a correct $\Ser_j$ will execute it before changing to the next view.
\item {\em $\Ser_i$ executed $op$ during view-change operation.} There are two possible cases:
\begin{enumerate}
\item {\em $\Ser_i$ executed $op$ before $\Ser_j$ changing to the next view.}
In this case, there are at least $f+1$ replicas that have committed to execute the history containing $op$ before $\Ser_j$ changing to the next view.
Since $\Ser_j$ needs to receive $f+1$ $\reqviewchange$ messages, there must be an intersection replica $\Ser_k$ that includes $op$ to its $\reqviewchange$ message.
Then, a correct $\Ser_j$ will execute $op$ before changing to the next view.
\item {\em $\Ser_i$ executed $op$ after $\Ser_j$ changing to the next view.}
Due to the same reason as case (a), $\Ser_i$ will process the same execution history (without $op$) as the one $\Ser_j$ executed. 
\end{enumerate}
\end{enumerate}

Therefore, we conclude that if a correct replica executes an operation $op$ in a view $v$, all correct replicas will execute $op$ before changing to a new view.
\end{proof}

\begin{theo}
{\em Let $seq = \langle op_1, ..., op_m\rangle$ be a sequence of operations executed by a correct replica $\Ser_i$,
then all other correct replicas executed the same sequence or a prefix of it.}
\end{theo}

\begin{proof}
Assume a correct replica $\Ser_j$ executed a sequence of operations $seq'$ that is not a prefix of $seq$,
i.e., there is at least one operation $op'_k$ that is different from $op_k$.
Assume that $op_k$ was executed in view $v$ and $op'_k$ was executed in view $v'$.
If $v' = v$, this contradicts Lemma 1, and if $v' \neq v$, this contradicts Lemma 2---thus proving the theorem.
\end{proof}

\subsection{Liveness}

We say that $\Cli$'s request {\em completes} when $\Cli$ accepts the reply.
We show that an operation requested by a correct $\Cli$ eventually completes.
We say a view is {\em stable} if the primary is correct.

\begin{lemma}
\label{liveness1}
{\em During a stable view, an operation $op$ requested by a correct client will complete.}
\end{lemma}

\begin{proof}
Since the primary $\Ser_p$ is correct, a valid $\prepare$ message will be sent.
If all active replicas behave correctly, the request will complete.
However, a faulty replica $\Ser_j$ may either crash or reply with a wrong share. 
This behavior will be detected by its parent (\fig{fig:normal_case}, line~\ref{28}) and $\Ser_j$ will be replaced by a passive replica (\fig{fig:normal_case}, line~\ref{21}).
\changedagain{If a threshold number of failure detections has been reached, correct replicas will initiate a view-change to switch to the fallback protocol.
The view-change will succeed since the primary is correct. In the fallback protocol, the request will complete as long as the number of non-primary faults is at most $f$.}
\end{proof}

\begin{lemma}
\label{liveness2}
{\em A view $v$ eventually will be changed to a stable view if $f + 1$ correct replicas request view-change.}
\end{lemma}

\begin{proof}
Suppose a quorum $Q$ of $f+1$ correct replicas requests a view-change.
We distinguish between three cases:
\begin{enumerate}
\item {\em The new primary $\Ser_{p'}$ is correct and all replicas in $Q$ received a valid} $\newview$ {message}. 
They will change to a stable view successfully (\fig{fig:view_change}, line~\ref{2.10}).

\item {\em None of the correct replicas received a valid} $\newview$ {\em message.} 
In this case, another view-change will start.

\item {\em Only a quorum $Q'$ of less than $f + 1$ correct replicas received a valid} $\newview$ {\em message.}
In this case, faulty replicas can follow the protocol to make the correct replicas in $Q'$ 
change to a non-stable view.
Other correct replicas 
will send new $\reqviewchange$ messages due to timeout, but a view-change will not start since they are less than $f+1$.
When faulty replicas deviate from the protocol, the correct replicas in $Q'$ will \mbox{trigger a new view-change.}

\end{enumerate}

In cases 2 and 3, a new view-change triggers the system to reach again one of the above three cases.
\changed{Recall that, under a weak synchrony assumption, messages are guaranteed to be delivered in polynomial time.
Therefore, the system will eventually reach case 1, i.e., a stable view \mbox{will be reached.}}
\end{proof}

\begin{theo}
{\em An operation requested by a correct client eventually completes.}
\end{theo}

\begin{proof}
In stable views, operations will complete eventually (Lemma~\ref{liveness1}).
If the view is not stable, there are two cases: 
\begin{enumerate}
\item {\em At least $f + 1$ correct replicas request a view-change.} The view will eventually be changed to stable \mbox{(Lemma 4).}
\item  {\em Less than $f + 1$ correct replicas request a view-change.} 
Requests will complete if all active replicas follow the protocol.
Otherwise, requests will not complete within a timeout, and eventually all correct replicas will request view-change and the system falls to case 1.
\end{enumerate}

Therefore, all replicas will eventually fall into a stable view and clients' requests will complete.
\end{proof}

\section{Design Choices}
\label{sec:discussion}

\subsection{\changed{Virtual Counter}}
\label{sec:counter}

Throughout the paper, we assume that each $\TA$ maintains a monotonic counter.
The simplest way to realize a monotonic counter is to directly use a hardware monotonic counter supported by the underlying TEE platform
(for example, MinBFT used TPM~\cite{TPM} counters and CheapBFT used counters realized in FPGA; Intel SGX platforms also support monotonic counters in hardware~\cite{SGXcounter}).
However, such hardware counters constitute a bottleneck for BFT protocols due to their low efficiency: for example, when using SGX counters, a read operation takes 60-140~ms and an increment operation takes~80-250 ms, depending on the platform~\cite{ROTE}.

An alternative is to have the TEE maintain a virtual counter in volatile memory; but it will be reset after each system reboot.
This can be naively solved by recording the counter value on persistent storage before reboot, but this solution suffers from the rollback attacks~\cite{ROTE}:
a faulty $\Ser_p$ can call the {\em request\_counter} function twice, each of which is followed by a machine reboot.
As a result, $\Ser_p$'s $\TA$ will record two counter values on the persistent storage.
$\Ser_p$ can just throw away the second value when the $\TA$ requests the latest backup counter value.
In this case, $\Ser_p$ can successfully equivocate.

To remedy this, we borrow the idea from~\cite{dirtybit}:
when $\TA$ wants to record its state (e.g., in preparation for a machine reboot), it increments its hardware counter $C$ and stores $(C+1, c, v)$ on persistent storage.
On reading back its state, the $\TA$ accepts the virtual counter value if and only if the current hardware counter value matches the stored one.
If the $\TA$ was terminated without incrementing and saving the hardware counter value (called {\em unscheduled reboot}),
it will find a mismatch and refuse to process any further requests from this point on.
This completely prevents equivocation; a faulty replica can only achieve DoS by causing unscheduled reboots.

In \BFT, we treat an unscheduled reboot as a crash failure. 
To bound the number of failures in the system, we provide a {\em reset\_counter} function to allow crashed (or rebooted) replicas to rejoin the system.
Namely, after an unscheduled reboot, $\Ser_i$ can broadcast a \rejoin message.
Replicas who receive this message will reply with a signed counter value together with the message log since the last checkpoint (similar to the \viewchange message).
$\Ser_i$'s $\TA$ can reset its counter value and work again
if and only if it receives $f+1$ consistent signed counter values from different replicas (line~\ref{1.49} in \fig{fig:tee}).
However, a faulty $\Ser_p$ can abuse this function to equivocate: request a signed counter value, enforce an unscheduled reboot, and then broadcast a \rejoin message to reset its counter value.
In this case, $\Ser_p$ can successfully associate two different messages with the same counter value.
To prevent this, we have all replicas refuse to provide a signed counter value to an unscheduled rebooted primary,
so that $\Ser_p$ can reset its counter value only when it becomes a normal \mbox{replica after a view-change.}

\subsection{BFT \`A la Carte}
\label{sec:alacarte}

\begin{figure*}[htbp]
    \centering
    \begin{subfigure}[Design choices (not all combinations are possible: e.g., X and C cannot be combined).]{
        \label{fig:components}
        \includegraphics[width=0.42\textwidth]{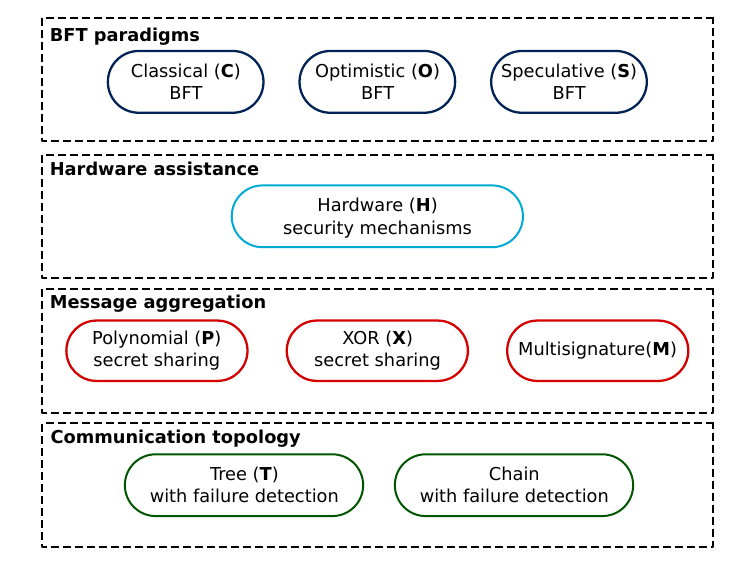}}
    \end{subfigure}
    \hfill
    \begin{subfigure}[Performance of some design choice combinations.]{
        \label{fig:comparisons}
        \includegraphics[width=0.42\textwidth]{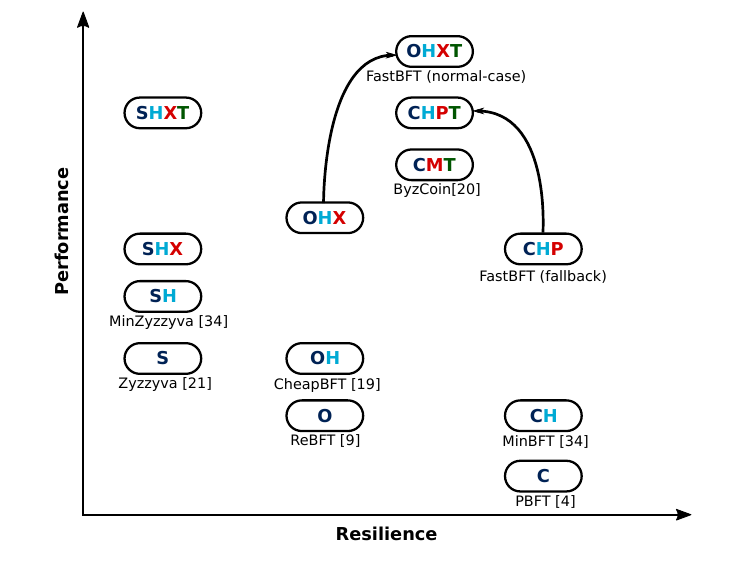}}
    \end{subfigure}
    \caption{Design choices for BFT protocols.}
\end{figure*}

In this section, we revisit our design choices in \BFT, show different protocols that can result from alternative design choices and qualitatively compare them \mbox{along two dimensions:}
\begin{itemize}
\item {\bf Performance:}  latency required to complete a request (lower the better) 
and the peak throughput (higher the better) of the system in common case. 
Generally (but not always), schemes that exhibit low latency also have high throughput; and
\item {\bf Resilience:} cost required to tolerate non-primary faults\footnote{\scriptsize{All BFT protocols require view-change to recover from primary faults,
which incurs a similar cost in different protocols.}}.
\end{itemize}

\fig{fig:components} depicts design choices for constructing BFT protocols; \fig{fig:comparisons} compares interesting combinations. Below, we discuss different possible BFT protocols, informally discuss their performance, resilience, and placement in \fig{fig:comparisons}.

\noindent \textbf{BFT paradigms.}
As mentioned in Section~\ref{sec:Background}, we distinguish between three possible paradigms:
classical (C) (e.g., PBFT~\cite{PBFT}), optimistic (O) \revision{(e.g., \textlabel{Distler et. al}{R2(5.3)}~\cite{Distler2016})}, and speculative (S) (e.g., Zyzzyva~\cite{Zyzzyva}).
Clearly, speculative BFT protocols (S) provide the best performance since it avoids all-to-all multicast.
However, speculative execution cannot tolerate even a single crash fault and requires clients' help to recover from inconsistent states.
In real-world scenarios, clients may have neither incentives nor resources (e.g., lightweight clients) to do so.
If a (faulty) client fails to report the inconsistency, replicas whose state has diverged from others may not discover this.
Moreover, if inconsistency appears, replicas may have to rollback some executions, which makes the programming model more complicated.
Therefore, speculative BFT fares the worst in terms of resilience.
In contrast, classical BFT protocols (C) can tolerate non-primary faults for free but requires all replicas to be involved in the agreement stage.
By doing so, these protocols achieve the best resilience but at the expense of bad performance.
Optimistic BFT protocols~(O) achieve a tradeoff between performance and resilience.
They only require active replicas to execute the agreement protocol which significantly reduces message complexity but still requires all-to-all multicast.
Although these protocols require transition~\cite{CheapBFT} or view-change~\cite{XFT} to tolerate non-primary faults, they require neither support from the clients nor any rollback mechanism. 

\noindent \textbf{Hardware assistance.}
Hardware security mechanisms (H) can be used in all three paradigms.
For instance, MinBFT~\cite{MinBFT} is a classical (C) BFT leveraging hardware security (H); to ease presentation, we say that MinBFT is of the CH family. Similarly, CheapBFT~\cite{CheapBFT} is OH (i.e., optimistic + hardware security) and MinZyzzyva~\cite{MinBFT} is SH (i.e., speculative + hardware security).
Hardware security mechanisms improve performance in all three paradigms (by reducing the number of required replicas and/or communication phases) without impacting resilience.

\noindent \textbf{Message aggregation.} 
We distinguish between message aggregation based on multisignatures (M)~\cite{CoSi} and on secret sharing (such as the one used in \BFT). We further classify secret sharing techniques into (the more efficient) XOR-based (X) and (the less efficient) polynomial-based (P).
Secret sharing techniques are only applicable to hardware-assisted BFT protocols (i.,e to CH, OH, and SH).
In the CH family, only polynomial-based secret sharing is applicable since classical BFT only requires responses from a threshold number of replicas in \commitP and \replyP.
\changedagain{Notice that CHP is the fallback protocol of \BFT.}
XOR-based secret sharing can be used in conjunction with \mbox{OH and SH}.
Message aggregation significantly increases performance of optimistic and classical BFT protocols but is of little help to speculative BFT which already has $O(n)$ message complexity.
After adding message aggregation, optimistic BFT protocols (OHX) become more efficient than speculative ones (SHX), since both of them have $O(n)$ message complexity but OHX requires less replicas to actively run the protocol.

\noindent \textbf{Communication topology.}
In addition, we can improve efficiency using better communication topologies (e.g., tree).
We can apply the tree topology with failure detection (T) to any of the above combinations e.g., CHPT, OHXT (which is \BFT), SHXT and CMT (which is ByzCoin~\cite{ByzCoin}).
Tree topology improves the performance of all protocols. 
For SHXT, resilience remains the same as before, since it still requires rollback in case of faults.
For OHXT, resilience will be improved, since transition or view-change is no longer required for non-primary faults.
On the other hand, for CHPT, resilience will almost be reduced to the same level as OHXT,
since a faulty node in the tree can make its whole subtree ``faulty'',
thus it can no longer tolerate non-primary faults for free. 
\changedagain{Chain is another communication topology widely used in BFT protocols~\cite{Next700, BChain}.
It offers high throughput but incurs large latency due to its $O(n)$ communication steps}.
Other communication topologies may provide better efficiency and/or resilience.
We leave the investigation and comparison of them \mbox{as future work}.

In \fig{fig:comparisons}, we summarize the above discussion visually. We conjecture that the use of hardware and the message aggregation can bridge the gap in performance between optimistic and speculative paradigms without adversely impacting resilience. The reliance on the tree topology further enhances performance and resilience. In the next section, we confirm these conjectures experimentally.

\remove{
\begin{figure*}[t]
    \centering
    \includegraphics[width=0.9\textwidth]{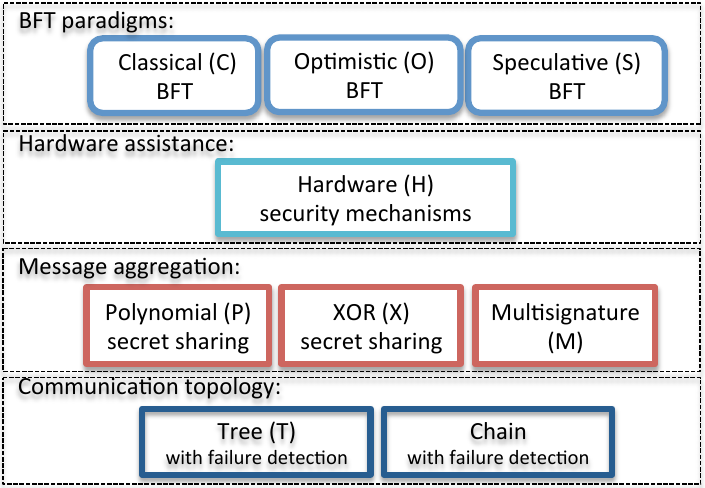} 
    \caption{Synthesis and comparisons of design choices.}
    \label{fig:synthesis}
    \emph{A note}
\end{figure*}
}

\remove{
\begin{SCfigure}
	\label{fig:modules}
        \includegraphics[scale=0.48]{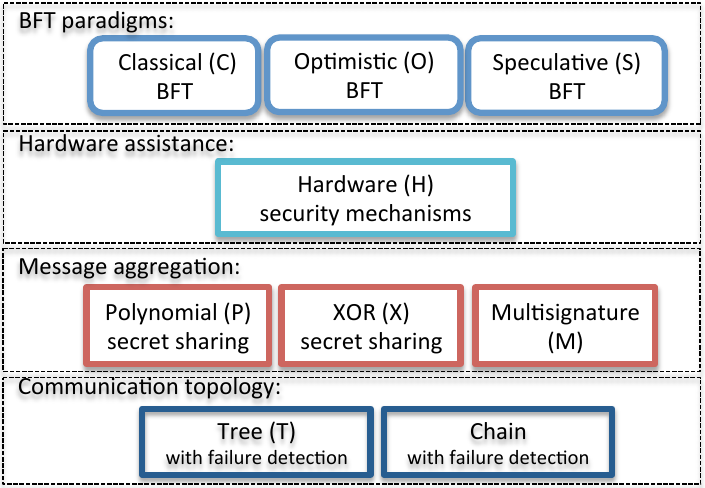}
        \caption{\footnotesize{Different combinations:} \\
 {\scriptsize Classic BFT+(1)=MinBFT \\
Optimistic BFT + (1) = CheapBFT \\
Speculative BFT + (1) = MinZyzzyva \\
MinBFT + (2.1) = MinAggre \\
CheapBFT + (2.2)  = CheapAggre \\
MinZyzzyva + (2.2) = MinZAggre \\
MinAggre + (3.1) = MinAggreTree  \\
CheapAggre + (3.1)= \BFT \\
MinZAggre + (3.1) = MinZAggreTree  \\
Classic BFT + (2.3) + (3.1) = CoSi \\}
}
\end{SCfigure}
\begin{figure}[t]
	\centering
        \centerline{\includegraphics[scale=0.57]{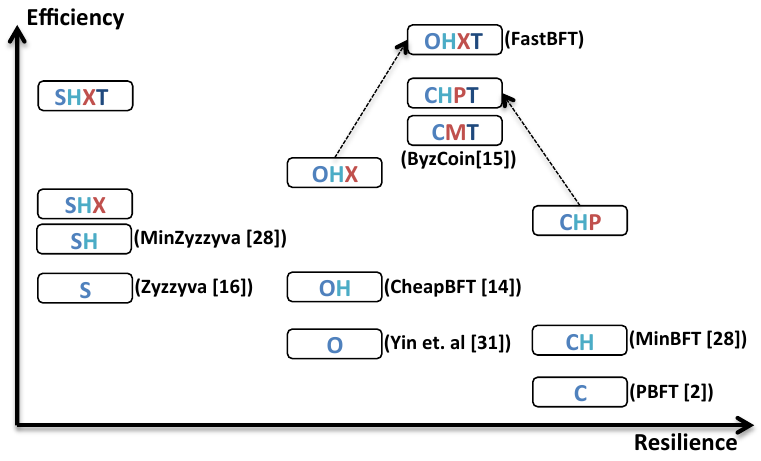}}
        \caption{Comparison between different combinations.}
        \label{fig:comparisons}
\end{figure}
}


\begin{figure}[b]
  \centering
   \includegraphics*[width=0.77\linewidth]{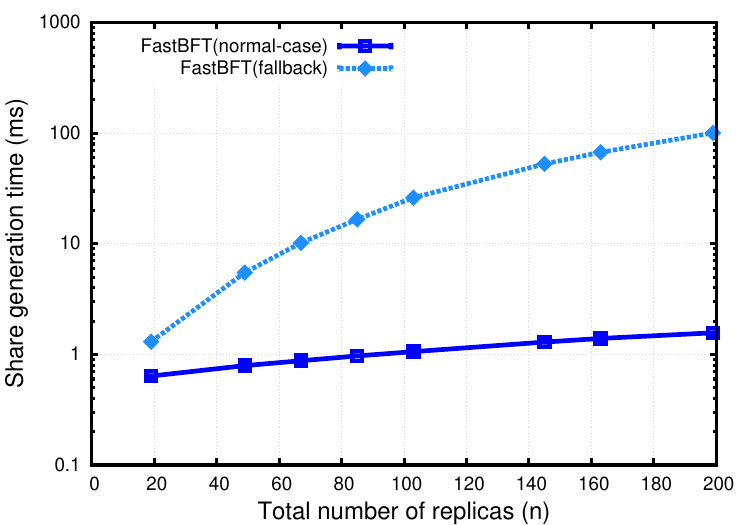}
   \caption{Cost of pre-processing vs. number of replicas ($n$)}
   \label{fig:preprocessing}
\end{figure}

\begin{figure}[t]
  \centering
   \includegraphics*[width=0.77\linewidth]{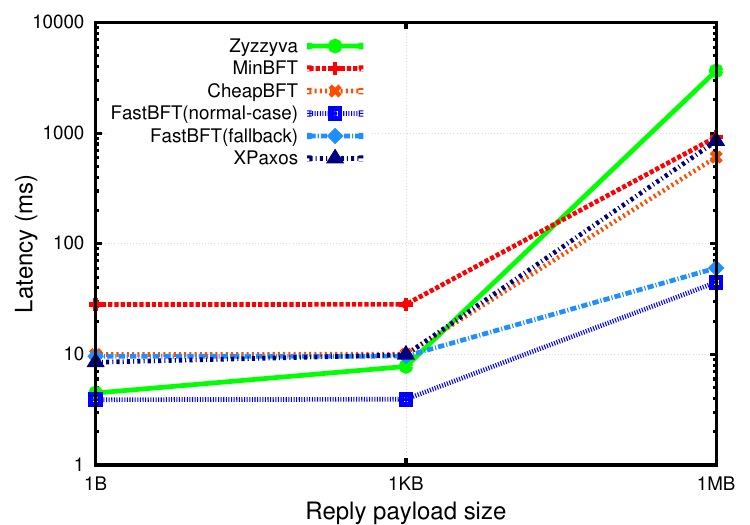}
   \caption{Latency vs. payload size.}
   \label{fig:payloads}
   \vspace{-0.3 em}
\end{figure}

\begin{figure*}[tb]
  \centering
      \subfigure[Peak throughput vs. $n$.
     ]{
          \label{fig:9}
          \includegraphics*[width=0.32\linewidth]{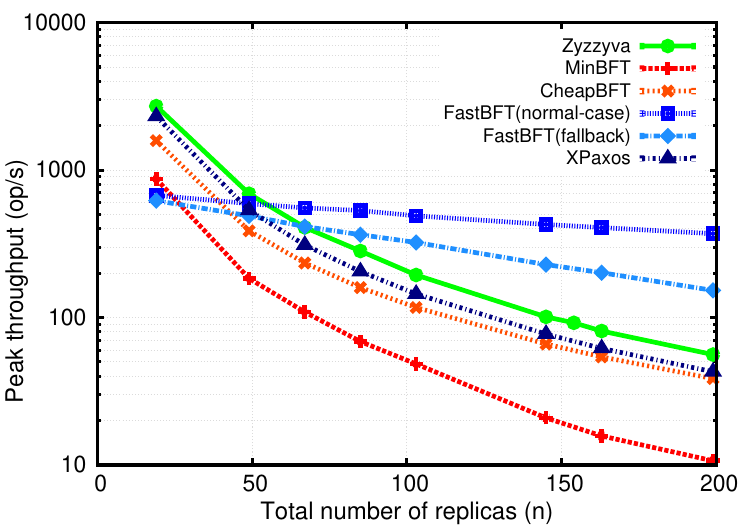}}
     \subfigure[Peak throughput vs. $f$.
     ]{
          \label{fig:2}
          \includegraphics*[width=0.32\linewidth]{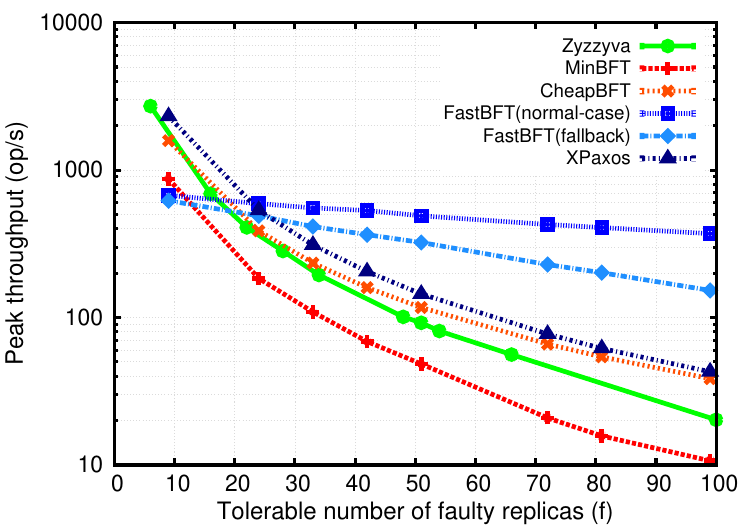}}
     \subfigure[Latency vs. $f$.
     ]{
          \label{fig:1}
          \includegraphics*[width=0.32\linewidth]{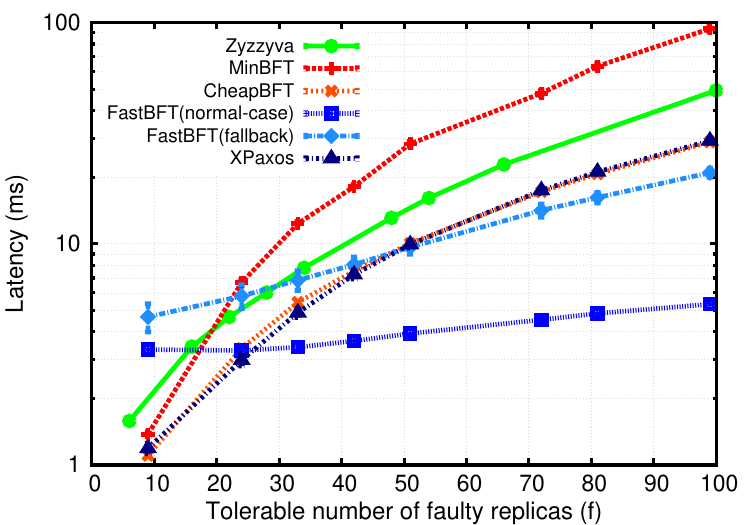}}
    \remove{ \subfigure[Latency vs. throughput for  $f$=$51$.
     ]{
          \label{fig:3}
          \includegraphics*[width=0.32\linewidth]{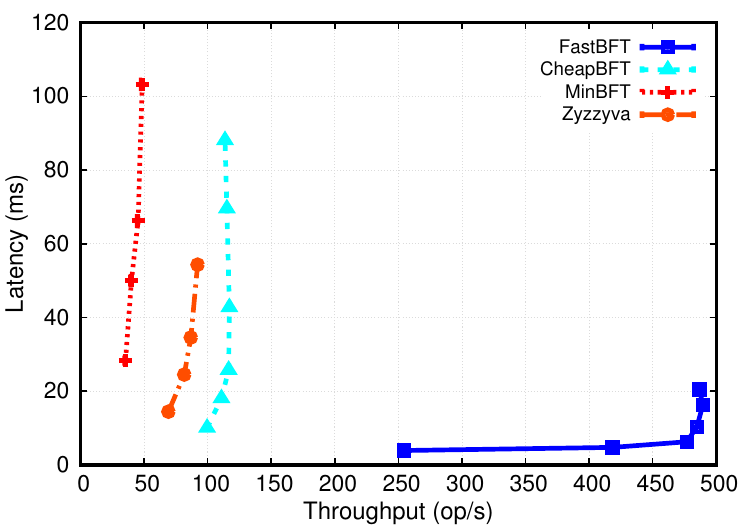}}}

     \caption{Evaluation results for 1 KB payload.}
    \label{fig:results}
\end{figure*}

\section{Evaluation}
\label{sec:performance}

In this section, we implement \BFT, \changed{emulating both the normal-case (cf. Section~\ref{sec:normal_case}) and the fallback protocol (cf. Section~\ref{sec:minBFT}),} and compare their performance with Zyzzyva~\cite{Zyzzyva}, MinBFT~\cite{MinBFT},  CheapBFT~\cite{CheapBFT} and \ICchanged{XPaxos~\cite{XFT}}.
Noticed that the fallback protocol is considered to be the worst-case of \BFT.

\subsection{Performance Evaluation: Setup and Methodology}

Our implementation is based on Golang. We use Intel SGX to provide hardware security support and implement the TEE part of a \BFT replica as an SGX enclave.
We use SHA256 for hashing, 128-bit CMAC for MACs, and 256-bit ECDSA for client signatures.
\changed{We set the size of the committed secret in \BFT to 128 bits and implement the monotonic counter as we described in Section~\ref{sec:counter}.}

We deployed our BFT implementations on a private network consisting of five 8 vCore Intel Xeon E3-1240 equipped with 32 GB RAM and Intel SGX.
All BFT replicas were running in separate processes. At all times, we load
balance the number of BFT \revision{\textlabel{replicas running}{E(4.1)}} on each machine; \revision{by varying the server failure threshold $f$ from 1 to 99, }we spawned a maximum of \revision{298 processes across} 5 machines. The clients were running on an 8 vCore Intel Xeon E3-1230 equipped with 16~GB RAM \revision{as multiple threads}. Each machine has 1~Gbps of bandwidth and the communication between various machines was bridged
using a 1 Gbps switch. \changed{\textlabel{This setup emulates a realistic enterprise deployment}{R2(2)}; for example IBM plans the deployment of their
blockchain platform within a large internal cluster~\cite{IBMblockchain}}, 
serving mutually distrustful parties (e.g., a consortium of banks using a cloud service for running a permissioned blockchain).

Each client invokes operation in a closed loop, i.e., each client may have at most one pending operation.
\revision{\textlabel{The latency of an operation is measured}{E(4.2)} as the time when a request is issued until the replicas' replies are accepted; and we define the throughput as the number of operations that can be handled by the system in one second.}
We evaluate the peak throughput with respect to the server failure threshold $f$. We also evaluate the latency incurred in the investigated BFT protocols with respect to the attained throughput.
We require that the clients issue back to back requests\revision{, i.e., \textlabel{a client issues the next request}{E(4.3)} as soon as the replies of the previous one have been accepted. }We then increase \revision{the concurrency by increasing} the number of clients in the system until
the aggregated throughput attained by all requests is saturated.
\revision{\textlabel{In our experiments, we vary the number of concurrent clients}{E(4.4)} from 1 to 10 to measure the latency and find the peak throughput.}
Note that each data point in our plots is averaged over 1,500 different measurements; where appropriate, we include the corresponding 95\% confidence intervals.

\subsection{Performance Evaluation: Results}

\begin{figure*}[tb]
  \centering
        \subfigure[Peak throughput vs. $n$.
     ]{
          \label{fig:10}
          \includegraphics*[width=0.32\linewidth]{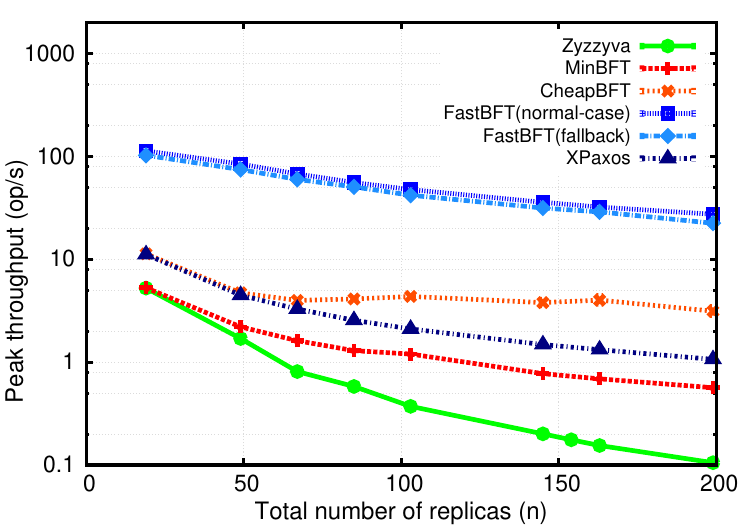}}
     \subfigure[Peak throughput vs. $f$.
     ]{
          \label{fig:5}
          \includegraphics*[width=0.32\linewidth]{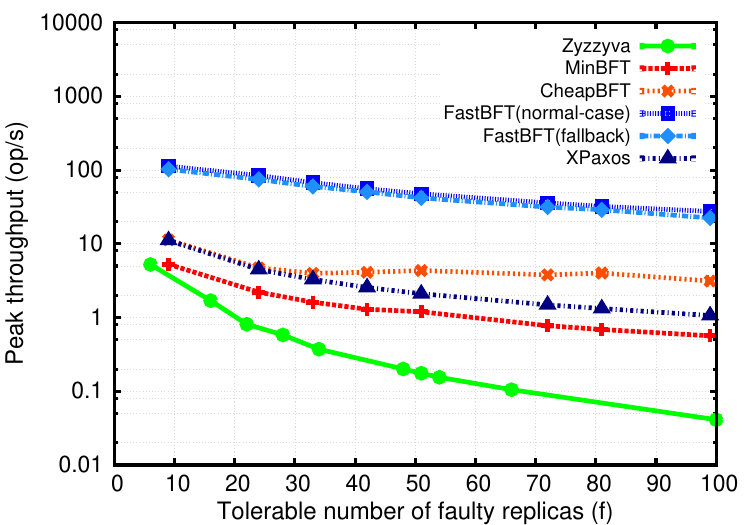}}
     \subfigure[Latency vs. $f$.
     ]{
         \label{fig:4}
         \includegraphics*[width=0.32\linewidth]{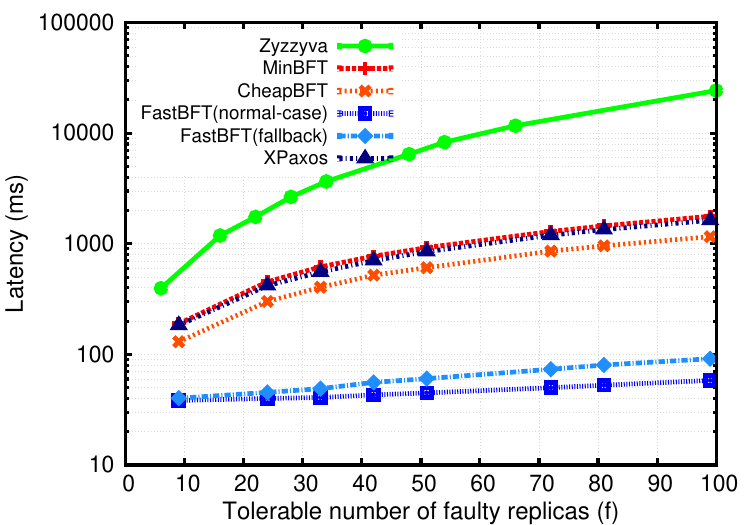}}
  \remove{\subfigure[Latency vs. throughput for $f$=$51$.
     ]{
          \label{fig:6}
          \includegraphics*[width=0.32\linewidth]{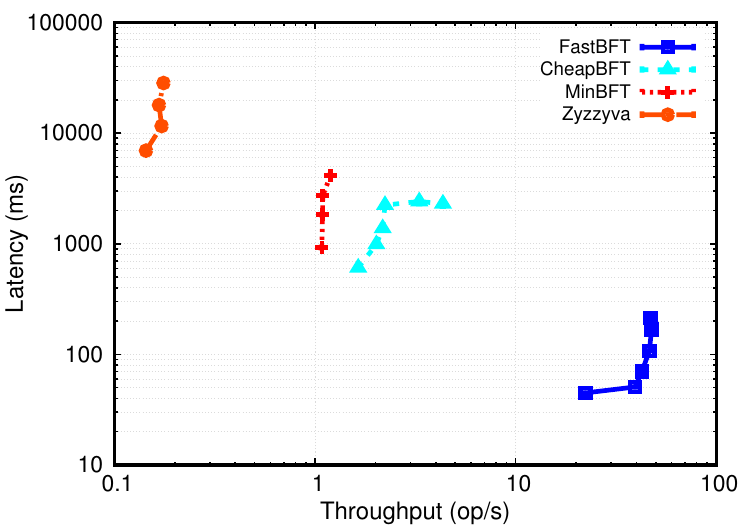}}}

     \caption{Evaluation results for 1 MB payload.}
    \label{fig:results}
\end{figure*}


\noindent\textbf{Pre-processing time.} \fig{fig:preprocessing} depicts the CPU time vs. number of replicas ($n$) measured when generating shares for one secret.
Our results show that in the normal case, TEE only spends about 0.6~ms to generate additive shares for 20 replicas; this time increases linearly as $n$ increases (e.g., 1.6~ms for 200 replicas).
This implies that it only takes several seconds to generate secrets for thousands of counters (queries).
We therefore argue that the preprocessing will not create a bottleneck for \BFT.
\changed{In the case of the fallback variant of \BFT, the share generation time (of Shamir secret shares) increases significantly as $n$ increases, 
since the process involves $n\cdot f$ modulo multiplications. 
Our results show that it takes approximately 100~ms to generate shares for 200 replicas. }
Next, we evaluate the online performance of \BFT.

\noindent \textbf{Impact of reply payload size.} We start by evaluating the latency vs. payload size (ranging from 1~byte to 1MB).
We set $n=103$ (which corresponds to our default network size).
\fig{fig:payloads} shows that \BFT achieves the lowest latency for all payload sizes.
For instance, to answer a request with 1~KB payload, \BFT requires 4 ms, which is twice as fast as Zyzzyva.
Our findings also suggest that the latency is mainly affected by payload sizes that are larger than 1~KB (e.g., 1~MB).
We speculate that this effect is caused by the overhead of transmitting large payloads. 
Based on this observation, we proceed to evaluate online performance for payload sizes of 1~KB and 1~MB respectively.
The payload size plays an important role in determining the effective transactional throughput of a system.
For instance, Bitcoin's consensus requires 600 seconds on average, but since payload size (block size) is 1~MB, Bitcoin can achieve a peak throughput of 7 transactions per second
(each Bitcoin transaction is \mbox{250 bytes on average).}

\noindent \textbf{Performance for 1KB reply payload.}
\fig{fig:9} depicts the peak throughput vs. $n$ for 1~KB payload.
\BFT's performance is modest when compared to other protocols when $n$ is small.
While the performance of these latter protocols degrades significantly as $n$ increases, \BFT's performance is marginally affected.
For example, when $n=199$, \BFT achieves a peak throughput of 370 operations per second when compared to 56, 38, \ICchanged{42} op/s for Zyzzyva, CheapBFT \ICchanged{and XPaxos} respectively.
\changed{Even in the fallback case, \BFT achieves almost 152 op/s when $n=199$ and outperforms the remaining protocols.}
Notice that comparing performance with respect to $n$ does not provide a fair basis to compare BFT protocols with and without hardware assistance.
For instance, when $n=103$, Zyzzyva can only tolerate at most $f=34$ faults, while \BFT, CheapBFT, and MinBFT can tolerate $f=51$.
We thus investigate how performance varies with the maximum number of tolerable faults in Figs.~\ref{fig:2} and~\ref{fig:1}. 
In terms of the peak throughput vs. $f$, the gap between \BFT and Zyzzyva is even larger.
For example, when $f=51$, it achieves a peak throughput of 490 operations per second, which is 5 times larger than Zyzzyva.
In general, \BFT achieves the highest throughput while exhibiting the lowest average latency per operation \changed{when $f>24$. The competitive advantage of \BFT (and its fallback variant) is even more pronounced as $f$ increases.
Although \BFT-fallback achieves comparable latency to CheapBFT, it achieves a considerably higher peak throughput. For example, when $f=51$, \BFT-fallback reaches 320 op/s when compared to 110 op/s for CheapBFT. This is due to the fact that \BFT exhibits considerably less communication complexity than CheapBFT.} 
\ICchanged{Furthermore, we emphasize that XPaxos~\cite{XFT} provides comparable performance to Paxos. 
So we conclude that \BFT even outperforms the crash fault-tolerant schemes.}

\noindent \textbf{Performance for 1MB reply payload.}
The superior performance of \BFT becomes more pronounced as the payload size increases since \BFT incurs very low communication overhead.
\fig{fig:10} shows that for 1MB payload, the peak throughput of \BFT outperforms the others even for small $n$,
and the gap keeps increasing as $n$ increases (260 times faster than Zyzzyva when $n=199$).
Figure~\ref{fig:5} and \ref{fig:4} show the same pattern as in the 1KB case when comparing \BFT and Zyzzyva for a given $f$ value.
\changed{We also notice that all other protocols beside \BFT exhibit significant performance deterioration when the payload size increases to 1~MB. 
For instance, when the system comprises 200 replicas, a client needs to wait for at least 100 replies (each 1MB in size) in MinBFT, CheapBFT \ICchanged{and XPaxos}, and 200 replies amounting to 200~MB in Zyzzyva. 
\BFT overcomes this limitation by requiring only the primary node to reply to the client.}
\revision{\textlabel{An alternative way to overcome this limitation}{R2(13)} is having the client specifies a single replica to return a full response. Other replicas only return a digest of the response. 
This optimisation affects the resilience when the designated replica is faulty.
Nevertheless, we still measured the response latencies of protocols with this optimisation and the results are shown in Figure~\ref{fig:single_full}.
The performance of \BFT remains the same since it only returns one value to the client. 
Even through the performance of other protocols have been significantly improved, \BFT (normal-case) still outperforms others.}

\begin{figure}[htbp]
  \centering
   \includegraphics*[width=0.77\linewidth]{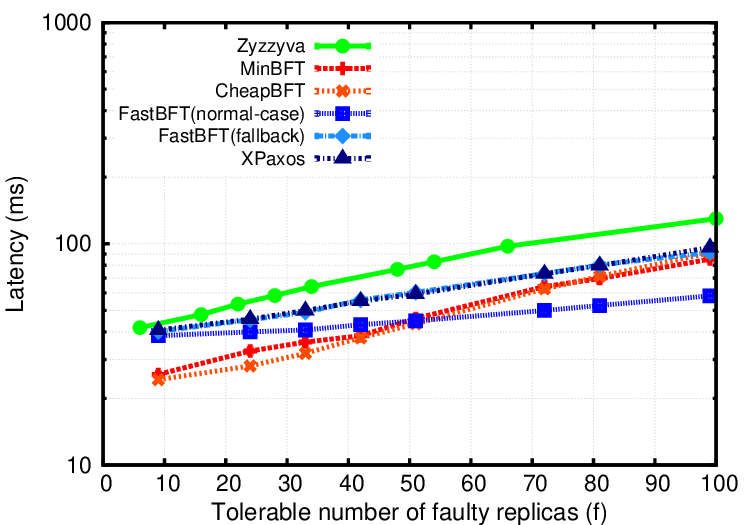}
   \caption{Latency vs.  $f$ (with single full-response)}
   \label{fig:single_full}
\end{figure}

%
%

Assuming that each payload comprises transactions of 250 bytes (similar to Bitcoin), \BFT can process a maximum of 113,246 transactions per second in a network of around 199 replicas.

Our results confirm our conjectures in Section~\ref{sec:discussion}: \BFT strikes a strong balance between performance and resilience.

\subsection{Security Considerations}

\ICchangedagain{\noindent \textbf{TEE usage.}
Since we assumed that TEEs may only crash (cf. system model in Section 3), 
a naive approach to implement a BFT protocol is to simply run a crash fault-tolerant variant (e.g., Paxos) within TEEs.}
\ICchanged{
However,  running large/complex code within TEEs increases the risk of vulnerabilities in the TEE code. 
The usual design pattern is to partition a complex application so that only a minimal, critical part runs within TEEs. 
Previous work (e.g., MinBFT, CheapBFT) showed that using minimal TEE functionality (maintaining an monotonic counter) improves the performance of BFT schemes. 
FastBFT presents a different way of leveraging TEEs that leads to significant performance improvements by slightly increasing the complexity of TEE functionality. 
FastBFT's TEE code has 7 interface primitives and 1,042 lines of code \revision{\textlabel{(47 lines of code are for SGX SDK)}{R1(6)}}; In comparison, MinBFT uses 2 interface functions and 191 lines \revision{(13 lines of code are for SGX SDK)} of code in our implementation. 
Both are small enough to make formal/informal verification as needed, \revision{ever though FastBFT places more functionality in the TEE than just a counter}.
In contrast, Paxos (based on LibPaxos~\cite{LibPaxos}) requires more than 4,000 lines of code.}

\noindent\ICchangedagain{\textbf{TEE side-channels.}}
\ICchanged{
SGX enclave code that deals with sensitive information must use side-channel resistant algorithms to process them~\cite{sidechannel}. 
However, the only sensitive information in FastBFT are cryptographic keys/secret-shares which are processed by standard cryptographic algorithms/implementations such as the standard the SGX crypto library (libsgx\_tcrypto.a) which are side-channel resistant. 
Existing side-channel attacks are based on either the RSA public component or the RSA implementation from other libraries, which we did not use in our implementation.}


\section{Related Work}
\label{sec:related_work}



{\em Randomized Byzantine consensus} protocols have been proposed in 1980s~\cite{Ben1983, Rabin1983}.
Such protocols rely on cryptographic coin tossing and expect to complete in $O(k)$ rounds with probability $1-2^{-k}$.
As such, randomized Byzantine protocols typically result in high communication and time complexities.
\changedagain{In this paper, we therefore focus on the efficient deterministic variants.}
Honeybadger~\cite{honeybadger} is a recent randomized Byzantine protocol that provides comparable throughput to PBFT.

Liu et al. observed that Byzantine faults are usually independent of asynchrony~\cite{XFT}.
Leveraging this observation, they introduced a new model, \emph{XFT}, which allows designing protocols that tolerate crash faults in weak synchronous networks and, meanwhile, tolerates Byzantine faults in synchronous network.
Following this model, the authors presented XPaxos, an optimistic
state machine replication, that requires $n=2f+1$ replicas to tolerate $f$ faults.
However, XPaxos still requires all-to-all multicast in the agreement stage---thus resulting \mbox{in $O(n^2)$ message complexity}. 

\BFT's message aggregation technique is similar to the {\em proof of writing} technique introduced in PowerStore~\cite{powerstore} which implements a read/write storage abstraction.
Proof of writing is a 2-round write procedure: the writer first commits to a random value, and then opens the commitment to ``prove'' that the first round has been completed.
The commitment can be implemented using cryptographic hashes or polynomial evaluation---thus removing the need for public-key operations.

\ICchanged{Hybster~\cite{Behl2017} is a TEE-based BFT protocol that leverages parallelization to improve performance, which is orthogonal to our contribution.}


\section{Conclusion and Future Work}
\label{sec:conclusion}

In this paper, we presented a new BFT protocol, \BFT. We analyzed and evaluated our proposal in comparison to existing BFT variants.
Our results show that \BFT 
is 6 times faster than Zyzzyva.
Since Zyzzyva reduces replicas' overheads to near their theoretical minima, we argue that \BFT achieves near-optimal efficiency for BFT protocols. 
Moreover, \BFT exhibits considerably slower decline in the achieved throughput as the network size grows when compared to other BFT protocols. 
This makes \BFT an ideal consensus layer candidate for \mbox{next-generation blockchain systems}.

We assume that TEEs are equipped with certified keypairs (Section~\ref{subsec:TEE}). Certification is typically done by the TEE manufacturer, but can also be done by any trusted party when the system is initialized. Although our implementation uses Intel SGX for hardware support, \BFT can be realized on any standard TEE platform (e.g., GlobalPlatform~\cite{GP}).

We plan to explore the impact of other topologies, besides trees, on the performance of \BFT. This will enable us to reason on
optimal (or near-optimal) topologies that suit a particular network size in \BFT.

\ifsubmission
\else
\section*{Acknowledgments}
The work was supported in part by a grant from NEC Labs Europe as well as funding from the Academy of Finland (BCon project, grant \#309195).
\fi
\label{sec:ack}

\raggedright
\bibliographystyle{IEEEtranS}
\bibliography{sigproc}


%

\begin{IEEEbiography}[{\includegraphics[width=1in,height=1.25in,clip,keepaspectratio]{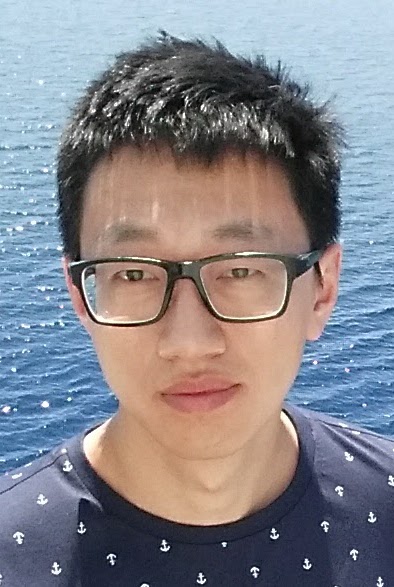}}]{Jian Liu}
is a Doctoral Candidate at Aalto University, Finland. 
He received his Masters of Science in University of Helsinki in 2014. 
He is instructed in applied cryptography and blockchains.
\end{IEEEbiography}

\begin{IEEEbiography}[{\includegraphics[width=1in,height=1.25in,clip,keepaspectratio]{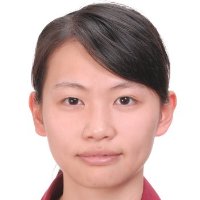}}]{Wenting~Li}
is a Senior Software Developer at NEC Laboratories Europe. She received her Masters of Engineering in Communication System Security from Telecom ParisTech in September 2011. She is interested in security with a focus on distributed system and IoT devices.
\end{IEEEbiography}

\begin{IEEEbiography}[{\includegraphics[width=1in,height=1.25in,clip,keepaspectratio]{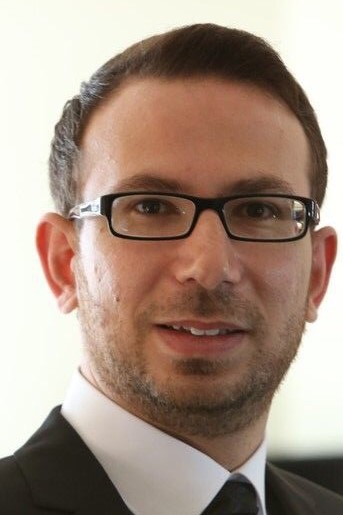}}]{Ghassan Karame}
is a Manager and Chief researcher of Security Group of NEC Laboratories Europe.
He received his Masters of Science from Carnegie Mellon University (CMU) in December 2006, 
and his PhD from ETH Zurich, Switzerland, in 2011. 
Until 2012, he worked as a postdoctoral researcher in ETH Zurich.
He is interested in all aspects of security and privacy with a focus on cloud security,
SDN/network security and Bitcoin security. He is a member of the IEEE and of the ACM. 
More information on his research at \url{http://ghassankarame.com/}.
\end{IEEEbiography}

\begin{IEEEbiography}[{\includegraphics[width=1in,height=1.25in,clip,keepaspectratio]{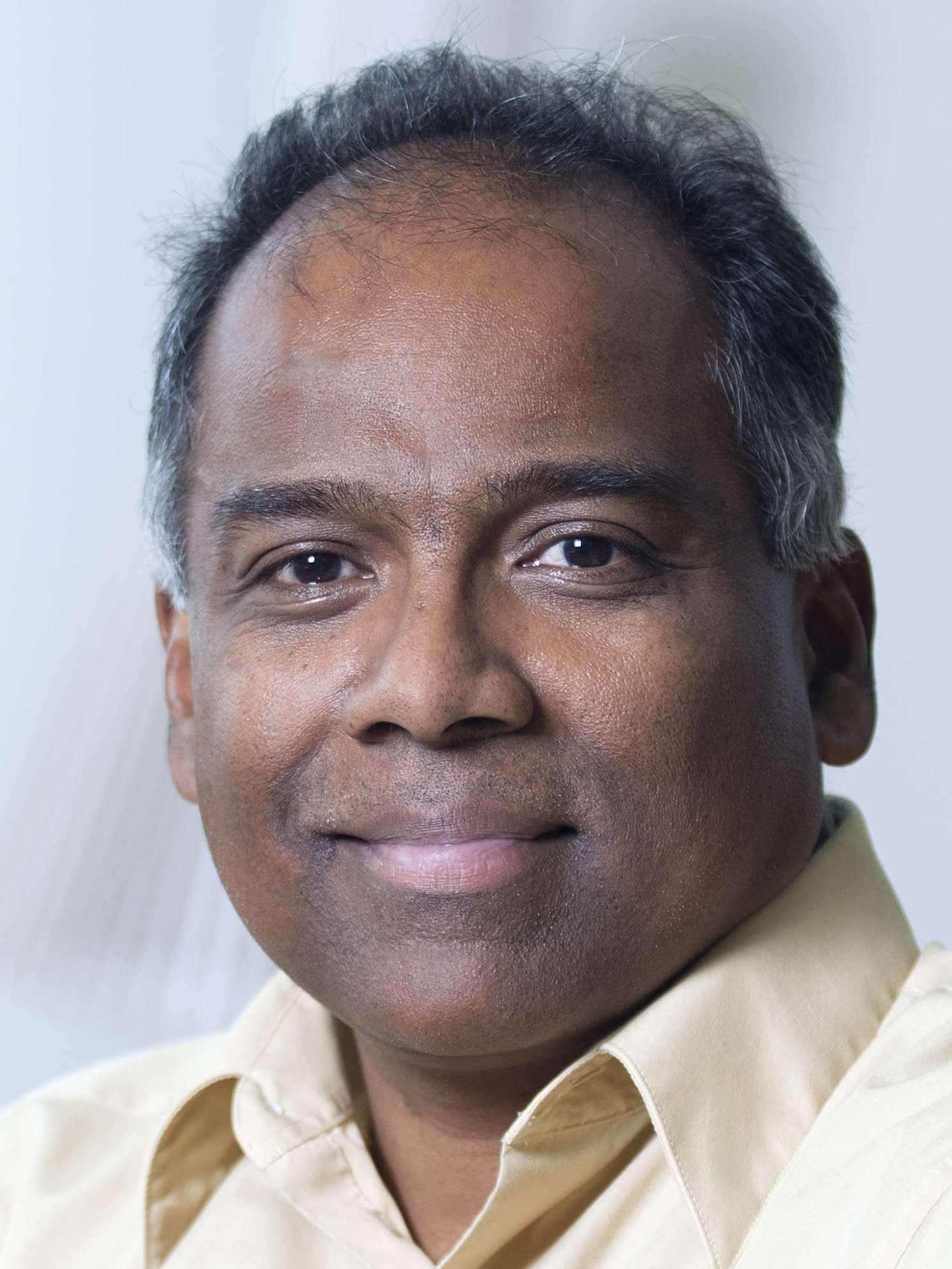}}]{N. Asokan} is a~Professor~of~Computer Science at Aalto University~where~he co-leads the secure systems~research group and directs 
Helsinki-Aalto~Center for Information Security -- HAIC.
More information on his research at \url{http://asokan.org/asokan/}.
\end{IEEEbiography}

\end{document}